\documentclass[sigconf]{aamas_hacked}

\usepackage{balance}

\RequirePackage{fix-cm}
\RequirePackage[T1]{fontenc}
\RequirePackage{xspace}
\RequirePackage{amsmath}
\RequirePackage{amsthm}
\RequirePackage{nicefrac}
\RequirePackage{tikz}
\RequirePackage{pgflibraryshapes}
\usetikzlibrary{positioning,arrows,shapes,snakes,calc,fit}
\RequirePackage{colortbl}
\RequirePackage{booktabs}
\RequirePackage{cleveref}
\RequirePackage{array}
\RequirePackage{graphicx}
\RequirePackage{xcolor}
\RequirePackage{enumitem}
\RequirePackage{tabularx}

\theoremstyle{definition}

\theoremstyle{plain}
\newtheorem{theorem}{Theorem}
\newtheorem{lemma}{Lemma}

\newtheoremstyle{bfnote}
{}{}
{}{}
{\bfseries}{.}
{ }
{\thmname{#1}\thmnumber{ #2}\thmnote{ (#3)}}
\theoremstyle{bfnote}
\newtheorem{remark}{Remark}

\newcommand{\rd}{\ensuremath{\mathit{RD}}\xspace}

\usepackage{thm-restate}
\usepackage{makecell}

\usetikzlibrary{patterns}

\newcommand{\Tau}{\mathrm{T}}

\theoremstyle{plain}

\setcopyright{ifaamas}
\acmConference[AAMAS '22]{Proc.\@ of the 21st International Conference
on Autonomous Agents and Multiagent Systems (AAMAS 2022)}{May 9--13, 2022}
{Online}{P.~Faliszewski, V.~Mascardi, C.~Pelachaud,
M.E.~Taylor (eds.)}
\copyrightyear{2022}
\acmYear{2022}
\acmDOI{}
\acmPrice{}
\acmISBN{}

\acmSubmissionID{261}

\title[Strategyproof Social Decision Schemes]{Relaxed Notions of Condorcet-Consistency and Efficiency for Strategyproof Social Decision Schemes}

\author{Felix Brandt}
\affiliation{
  \institution{Technical University of Munich}
  \city{Munich}
  \country{Germany}}
\email{brandtf@in.tum.de}

\author{Patrick Lederer}
\affiliation{
  \institution{Technical University of Munich}
  \city{Munich}
  \country{Germany}}
\email{ledererp@in.tum.de}

\author{Ren\'e Romen}
\affiliation{
  \institution{Technical University of Munich}
  \city{Munich}
  \country{Germany}}
\email{rene.romen@tum.de}

\begin{abstract}
Social decision schemes (SDSs) map the preferences of a group of voters over some set of $m$ alternatives to a probability distribution over the alternatives.
A seminal characterization of strategyproof SDSs by \citeauthor{Gibb77a} implies that there are no strategyproof Condorcet extensions and that only random dictatorships satisfy \emph{ex post} efficiency and strategyproofness. The latter is known as the \emph{random dictatorship theorem}. We relax Condorcet-consistency and \emph{ex post} efficiency by introducing a lower bound on the probability of Condorcet winners and an upper bound on the probability of Pareto-dominated alternatives, respectively. 
We then show that the SDS that assigns probabilities proportional to Copeland scores is the only anonymous, neutral, and strategyproof SDS that can guarantee the Condorcet winner a probability of at least $2/m$. Moreover, no strategyproof SDS can exceed this bound, even when dropping anonymity and neutrality. Secondly, we prove a continuous strengthening of Gibbard's random dictatorship theorem: the less probability we put on Pareto-dominated alternatives, the closer to a random dictatorship is the resulting SDS.
Finally, we show that the only anonymous, neutral, and strategyproof SDSs that maximize the probability of Condorcet winners while minimizing the probability of Pareto-dominated alternatives are mixtures of the uniform random dictatorship and the randomized Copeland rule.

\end{abstract}

\allowdisplaybreaks

\keywords{Randomized Social Choice; Social Decision Schemes; Strategyproofness; Condorcet-consistency; \emph{ex post} efficiency}
         
\newcommand{\BibTeX}{\rm B\kern-.05em{\sc i\kern-.025em b}\kern-.08em\TeX}

\begin{document}

\pagestyle{fancy}
\fancyhead{}

\maketitle 

\section{Introduction}

Multi-agent systems are often faced with problems of collective decision making: how to find a group decision given the preferences of multiple individual agents. These problems, which have been traditionally studied by economists and mathematicians, are of increasing interest to computer scientists who employ the formalisms of social choice theory to analyze computational multi-agent systems \citep[see, e.g.,][]{NRTV07a,ShLe08a,BCE12a,BCE+14a}.

A pervasive phenomenon in collective decision making is strategic manipulation: voters may be better off by lying about their preferences than reporting them truthfully. This is problematic since all desirable properties of a voting rule are in doubt when voters act dishonestly. 
Thus, it is important that voting rules incentivize voters to report their true preferences.
Unfortunately, \citet{Gibb73a} and \citet{Satt75a} have shown independently that dictatorships are the only non-imposing voting rules that are immune to strategic manipulations.
However, these voting rules are unacceptable for most applications because they invariably return the most preferred alternative of a fixed voter. A natural question is whether more positive results can be obtained when allowing for randomization. \citet{Gibb77a} hence introduced \emph{social decision schemes} (SDSs), which map the preferences of the voters to a lottery over the alternatives and defined SDSs to be \emph{strategyproof} if no voter can obtain more expected utility for any utility representation that is consistent with his ordinal preference relation. He then gave a complete characterization of strategyproof SDSs in terms of convex combinations of two types of restricted SDSs, so-called unilaterals and duples. An important consequence of this result is the \emph{random dictatorship theorem}: random dictatorships are the only \emph{ex post} efficient and strategyproof SDSs. Random dictatorships are convex combinations of dictatorships, i.e., each voter is selected with some fixed probability and the top choice of the chosen voter is returned.
In contrast to deterministic dictatorships, the uniform random dictatorship, in which every agent is picked with the same probability, enjoys a high degree of fairness and is in fact used in many subdomains of social choice
\citep[see, e.g.,][]{AbSo98a,ChKo10a}. As a consequence of these observations, \citeauthor{Gibb77a}'s theorem has been the point of departure for a lot of follow-up work. In addition to several alternative proofs of the theorem \citep[e.g.,][]{Dugg96a,Nand97a,Tana03a},
there have been extensions with respect to manipulations by groups \citep{Barb79a}, cardinal preferences \citep[e.g.,][]{Hyll80a,DPS07a,Nand12a}, weaker notions of strategyproofness \citep[e.g.,][]{Beno02a,Sen11a,ABBB15a,BBEG16a}, and restricted domains of preference \citep[e.g.,][]{DPS02a,CSZ14a}.

Random dictatorships suffer from the disadvantage that they do not allow for compromise. For instance, suppose that voters strongly disagree on the best alternative, but have a common second best alternative. In such a scenario, it seems reasonable to choose the second best alternative but random dictatorships do not allow for this compromise.
On a formal level, this observation is related to the fact that random dictatorships violate \emph{Condorcet-consistency}, which demands that an alternatives that beats all other alternatives in pairwise majority comparisons should be selected.
Motivated by this observation, we analyze the limitations of strategyproof SDSs by relaxing two classic conditions: Condorcet-consistency and \emph{ex post} efficiency. To this end, we say that an SDS is \emph{$\alpha$-Condorcet-consistent} if a Condorcet winner always receives a probability of at least $\alpha$ and \emph{$\beta$-{ex post} efficient} if a Pareto-dominated alternative always receives a probability of at most $\beta$. Moreover, we say a strategyproof SDS is \emph{$\gamma$-randomly dictatorial} if it can be represented as a convex combination of two strategyproof SDSs, one of which is a random dictatorship that will be selected with probability $\gamma$. All of these axioms are discussed in more detail in \Cref{subSec:Relax}.

Building on an alternative characterization of strategyproof SDSs by \citeauthor{Barb79b} \citep{Barb79b}, we then show the following results ($m$ is the number of alternatives and $n$ the number of voters):

\begin{itemize}
\item Let $m,n\ge 3$. There is no strategyproof SDS that satisfies $\alpha$-Condorcet-consistency for $\alpha > \nicefrac{2}{m}$. Moreover, the \emph{randomized Copeland rule}, which assigns probabilities proportional to Copeland scores, is the only strategyproof SDS that satisfies anonymity, neutrality, and $\nicefrac{2}{m}$-Condorcet-consistency.
\item Let $0 \le \epsilon \le 1$ and $m\geq 3$. Every strategyproof SDS that is $\frac{1 - \epsilon}{m}$-\emph{ex post} efficient is $\gamma$-randomly dictatorial for $\gamma \ge \epsilon$. If we require additionally anonymity, neutrality, and $m \ge 4$, then only mixtures of the uniform random dictatorship and the uniform lottery satisfy this bound tightly. 
\item Let $m \ge 4$ and $n \ge 5$. No strategyproof SDS that is $\alpha$-Condorcet-consistent is $\beta$-\emph{ex post} efficient for $\beta< \frac{m-2}{m-1}\alpha$. If we additionally require anonymity and neutrality, then only mixtures of the uniform random dictatorship and the randomized Copeland rule satisfy $\beta = \frac{m-2}{m-1}\alpha$.
\end{itemize}

The first statement characterizes the randomized Copeland rule as the ``most Condorcet-consistent'' SDS that satisfies strategyproofness, anonymity, and neutrality. In fact, no strategyproof SDS can guarantee more than $\nicefrac{2}{m}$ probability to the Condorcet winner, even when dropping anonymity and neutrality.
The second point can be interpreted as a continuous strengthening of \citeauthor{Gibb77a}'s random dictatorship theorem: the less probability we put on Pareto-dominated alternatives, the more randomly dictatorial is the resulting SDS. In particular,
this theorem indicates that we cannot find appealing strategyproof SDSs by allowing that Pareto-dominated alternatives gain a small probability since the resulting SDS will be very similar to random dictatorships. 
The last statement identifies a tradeoff between $\alpha$-Condorcet-consistency and $\beta$-\emph{ex post} efficiency: the more probability a strategyproof SDS guarantees to the Condorcet winner, the less efficient it is. Thus, we can either only maximize the $\alpha$-Condorcet-consistency or the $\beta$-\emph{ex post} efficiency of a strategyproof SDS, which again highlights the central roles of the randomized Copeland rule and random dictatorships.

\section{The model}

Let $N = \{1, 2, \dots, n\}$ be a finite set of voters and let $A = \{a, b, \dots\}$ be a finite set of $m$ alternatives. Every voter $i$ has a \emph{preference relation}~$\succ_i$, which is an anti-symmetric, complete, and transitive binary relation on $A$. We write $x\succ_i y$ if voter $i$ prefers $x$ strictly to $y$ and $x\succeq_i y$ if $x \succ_i y$ or $x = y$. The set of all preference relations is denoted by $\mathcal{R}$. A \emph{preference profile} $R \in \mathcal{R}^n$ contains the preference relation of each voter $i \in N$.
We define the \emph{supporting size} for $x$ against $y$ in the preference profile $R$ by $n_{x y}(R)=|\{i \in N: x \succ_i y\}|$.

Given a preference profile, we are interested in the winning chance of each alternative. We therefore analyze social decision schemes (SDSs), which map each preference profile to a lottery over the alternatives. A \emph{lottery} $p$ is a probability distribution over the set of alternatives $A$, i.e., it assigns each alternative $x$ a probability $p(x) \ge 0$ such that $\sum_{x \in A} p(x) = 1$. The set of all lotteries over $A$ is denoted by $\Delta(A)$. Formally, a \emph{social decision scheme (SDS)} is a function $f:\mathcal{R}^n\rightarrow \Delta(A)$. We denote with $f(R, x)$ the probability assigned to alternative $x$ by $f$ for the preference profile $R$. 

Since there is a huge number of SDSs, we now discuss axioms formalizing desirable properties of these functions. Two basic fairness conditions are anonymity and neutrality. Anonymity requires that voters are treated equally. Formally, an SDS $f$ is \emph{anonymous} if $f(R)=f(\pi(R))$ for all preference profiles $R$ and permutations $\pi:N\rightarrow N$. Here, $R'=\pi(R)$ denotes the profile with ${\succ_{\pi(i)}'} = {\succ_i}$ for all voters $i\in N$. \emph{Neutrality} guarantees that alternatives are treated equally and formally requires for an SDS $f$ that $f(R, x)=f(\tau(R), \tau(x))$ for all preference profiles $R$ and permutations $\tau:A\rightarrow A$. This time, $R'=\tau(R)$ is the profile derived by permuting the alternatives in $R$ according to $\tau$, i.e, $\tau(x)\succ_i' \tau(y)$ if and only if $x\succ_i y$ for all alternatives $x, y \in A$ and voters $i\in N$.

\subsection{Stochastic Dominance and Strategyproofness}\label{subsec:strategyproofness}

This paper is concerned with strategyproof SDSs, i.e., social decision schemes in which voters cannot benefit by lying about their preferences. In order to make this formally precise, we need to specify how voters compare lotteries. To this end, we leverage the well-known notion of stochastic dominance: a voter $i$ (weakly) prefers a lottery $p$ to another lottery $q$, written as $p \succeq_i q$, if $\sum_{y\in A:y \succ_i x} p(y) \geq\sum_{y\in A:y \succ_i x} q(y)$ for every alternative $x\in A$. Less formally, a voter prefers a lottery $p$ weakly to a lottery $q$ if, for every alternative $x\in A$, $p$ returns a better alternative than $x$ with as least as much probability as $q$. Stochastic dominance does not induce a complete order on the set of lotteries, i.e., there are lotteries $p$ and $q$ such that a voter $i$ neither prefers $p$ to $q$ nor $q$ to $p$.

Based on stochastic dominance, we can now formalize strategyproofness. An SDS $f$ is \emph{strategyproof} if $f(R)\succeq_i f(R')$ for all preference profiles $R$ and $R'$ and voters $i\in N$ such that ${\succ_j}={\succ_j'}$ for all $j\in N\setminus \{i\}$. Less formally, strategyproofness requires that every voter prefers the lottery obtained by voting truthfully to any lottery that he could obtain by voting dishonestly. Conversely, we call an SDS $f$ \emph{manipulable} if it is not strategyproof. While there are other ways to compare lotteries with each other, stochastic dominance is the most common one \citep[see, e.g,][]{Gibb77a,Barb79b,BoMo01a,EPS02a,ABBB15a}. This is mainly due to the fact that $p \succeq_i q$ implies that the expected utility of $p$ is at least as high as the expected utility of $q$ for every vNM utility function that is ordinally consistent with voter $i$'s preferences. Hence, if an SDS is strategyproof, no voter can manipulate regardless of his exact utility function \citep[see, e.g.,][]{Sen11a,BBEG16a}. This observation immediately implies that the \emph{convex combination} $h=\lambda f+(1-\lambda)g$ (for some $\lambda\in[0,1]$) of two strategyproof SDSs $f$ and $g$ is again strategyproof: a manipulator who obtains more expected utility from $h(R')$ than $h(R)$ prefers $f(R')$ to $f(R)$ or $g(R')$ to $g(R)$. 

\citet{Gibb77a} shows that every strategyproof SDS can be represented as convex combinations of unilaterals and duples.\footnote{In order to simplify the exposition, we slightly modified \citeauthor{Gibb77a}'s terminology by requiring that duples and unilaterals have to be strategyproof.} The terms ``unilaterals'' and ``duples'' refer here to special classes of SDSs: a \emph{unilateral} is a strategyproof SDS that only depends on the preferences of a single voter $i$, i.e., $f(R)=f(R')$ for all preference profiles $R$ and $R'$ such that ${\succ_i}={\succ_i'}$. A \emph{duple}, on other hand, is a strategyproof SDS that only chooses between two alternatives $x$ and $y$, i.e., $f(R, z)=0$ for all preference profiles $R$ and alternatives $z\in A\setminus \{x, y\}$.

\begin{theorem}[\citet{Gibb77a}]\label{thm:Gibbard1}
	An SDS is strategyproof if and only if it can be represented as a convex combination of unilaterals and duples.
\end{theorem}

Since duples and unilaterals are by definition strategyproof, \Cref{thm:Gibbard1} only states that strategyproof SDSs can be decomposed into a mixture of strategyproof SDSs, each of which must be of a special type. In order to circumvent this restriction, Gibbard proves another characterization of strategyproof SDSs. 

\begin{theorem}[\citet{Gibb77a}]\label{thm:Gibbard2}
	An SDS is strategyproof if and only if it is non-perverse and localized.
\end{theorem}

Non-perversity and localizedness are two axioms describing the behavior of an SDS. 
For defining these axioms, we denote with $R^{i:yx}$ the profile derived from $R$ by only reinforcing $y$ against $x$ in voter $i$'s preference relation. Note that this requires that $x\succ_i y$ and that there is no alternative $z\in A$ such that $x\succ_i z\succ_iy$. Then, an SDS $f$ is \emph{non-perverse} if $f(R^{i:yx},y)\geq f(R,y)$ for all preference profiles $R$, voters $i\in N$, and alternatives $x,y\in A$. Moreover, an SDS is \emph{localized} 
if $f(R^{i:yx},z)=f(R,z)$ for all preference profiles $R$, voters $i\in N$, and distinct alternatives $x,y,z\in A$. Intuitively, non-perversity---which is now often referred to as monotonicity---requires that the probability of an alternative only increases if it is reinforced, and localizedness that the probability of an alternative does not depend on the order of the other alternatives. Together, \Cref{thm:Gibbard1} and \Cref{thm:Gibbard2} show that each strategyproof SDS can be represented as a mixture of unilaterals and duples, each of which is non-perverse and localized.

Since Gibbard's results can be quite difficult to work with, we now state another characterization of strategyproof SDSs due to \citet{Barb79b}. This author has shown that every strategyproof SDS that satisfies anonymity and neutrality can be represented as a convex combination of a supporting size SDS and a point voting SDS. A \emph{point voting SDS} is defined by a scoring vector $(a_1, a_2, \dots, a_m)$ that satisfies $a_1 \ge a_2 \ge \dots \ge a_m \ge 0 $ and $\sum_{i \in \{1, \dots, m\}} a_i = \frac{1}{n}$. The probability assigned to an alternative $x$ by a point voting SDS $f$ is $f(R, x) = \sum_{i \in N} a_{|\{y \in A:y \succeq_i x \}|}$.
Furthermore, \emph{supporting size SDSs} also rely on a scoring vector $(b_n, b_{n-1}, \dots, b_0)$ with $b_n \ge b_{n-1} \ge \dots \ge b_0 \ge 0 $ and $b_i + b_{n-i} = \frac{2}{m(m-1)}$ for all $i\in \{0,\dots,n\}$ to compute the outcome. The probability assigned to an alternative $x$ by a supporting size SDS $f$ is then $f(R, x) = \sum_{y \in A \setminus \{x\}} b_{n_{xy}(R)}$.
Note that point voting SDSs can be seen as a generalization of (deterministic) positional scoring rules and supporting size SDSs can be seen as a variant of Fishburn's C2 functions \citep{Fish77a}.

\begin{theorem}[\citet{Barb79b}]\label{thm:Barbera}
	An SDS is anonymous, neutral, and strategyproof if and only if it can be represented as a convex combination of a point voting SDS and a supporting size SDS.
\end{theorem}

Many well-known SDSs can be represented as point voting SDSs or supporting size SDSs. For example, the \emph{uniform random dictatorship} $f_{\mathit{RD}}$, which chooses one voter uniformly at random and returns his best alternative, is the point voting SDS defined by the scoring vector $\left(\frac{1}{n}, 0, \dots, 0\right)$. An instance of a supporting size SDS is the \emph{randomized Copeland rule} $f_{C}$, which assigns probabilities proportional to the Copeland scores $c(x,R)=|\{y\in A\setminus \{x\}\colon n_{xy}(R)>n_{yx}(R)\}| + \frac{1}{2}|\{y\in A\setminus \{x\}\colon n_{xy}(R)=n_{yx}(R)\}|$. This SDS is the supporting size SDS defined by the vector $b=\left(b_n, b_{n-1}, \dots, b_0\right)$, where $b_i=\frac{2}{m(m-1)}$ if $i>\frac{n}{2}$, $b_i=\frac{1}{m(m-1)}$ if $i=\frac{n}{2}$, and $b_i=0$ otherwise. Furthermore, there are SDSs that can be represented both as point voting SDSs and supporting size SDSs. An example is the \emph{randomized Borda rule} $f_B$, which randomizes proportional to the Borda scores of the alternatives. This SDS is the point voting SDS defined by the vector $\left(\frac{2(m-1)}{nm(m-1)}, \frac{2(m-2)}{nm(m-1)}, \cdots, \frac{2}{nm(m-1)},0\right)$ and equivalently the supporting size SDS defined by the vector $ \left(\frac{2n}{nm(m-1)}, \frac{2(n-1)}{nm(m-1)}, \cdots, \frac{2}{nm(m-1)}, 0\right) $. Both the randomized Copeland rule and the randomized Borda rule were rediscovered several times by authors who were apparently unaware of \citeauthor{Barb79b}'s work \citep[see][]{Heck03a,CoSa06c,Proc10a,HeCh13a}.

\subsection{Relaxing Classic Axioms}\label{subSec:Relax}

The goal of this paper is to identify attractive strategyproof SDSs other than random dictatorships by relaxing classic axioms from social choice theory. In more detail, we investigate how much probability can be guaranteed to Condorcet winners and how little probability must be assigned to Pareto-dominated alternatives by strategyproof SDSs. In the following we formalize these ideas using $\alpha$-Condorcet-consistency and $\beta$-\emph{ex post} efficiency.

Let us first consider $\beta$-\emph{ex post} efficiency, which is based on Pareto-dominance. An alternative $x$ \emph{Pareto-dominates} another alternative $y$ in a preference profile $R$ if $x\succ_i y$ for all $i\in N$. 
The standard notion of \emph{ex post efficiency} then formalizes that Pareto-dominated alternatives should have no winning chance, i.e., $f(R,x)=0$ for all preference profiles $R$ and alternatives $x$ that are Pareto-dominated in $R$. As first shown by Gibbard, random dictatorships are the only strategyproof SDSs that satisfy \emph{ex post} efficiency. These SDSs choose each voter with a fixed probability and return his best alternative as winner. However, this result breaks down once we allow that Pareto-dominated alternatives can have a non-zero chance of winning $\beta>0$. For illustrating this point, consider a random dictatorship $d$ and another strategyproof SDS $g$. Then, the SDS $f^*=(1-\beta) d + \beta g$ is strategyproof for every $\beta\in (0,1]$ and no random dictatorship, but assigns a probability of at most $\beta$ to Pareto-dominated alternatives. We call the last property $\beta$-\emph{ex post} efficiency: an SDS $f$ is \emph{$\beta$-ex post efficient} if $f(R,x)\le\beta$ for all preference profiles $R$ and alternatives $x$ that are Pareto-dominated in $R$. 

A natural generalization of the random dictatorship theorem is to ask which strategyproof SDSs satisfy $\beta$-\emph{ex post} efficiency for small values of $\beta$. If $\beta$ is sufficiently small, $\beta$-\emph{ex post} efficiency may be quite acceptable. As we show, the random dictatorship theorem is quite robust in the sense that all SDSs that satisfy $\beta$-\emph{ex post} efficiency for $\beta<\frac{1}{m}$ are similar to random dictatorships. In order to formalize this observation, we introduce $\gamma$-randomly dictatorial SDSs: a strategyproof SDS $f$ is \emph{$\gamma$-randomly dictatorial} if $\gamma\in [0,1]$ is the maximal value such that $f$ can be represented as $f=\gamma d + (1-\gamma) g$, where $d$ is a random dictatorship and $g$ is another strategyproof SDS. In particular, we require that $g$ is strategyproof as otherwise, SDSs that seem ``non-randomly dictatorial'' are not $0$-randomly dictatorial. For instance, the uniform lottery $f_U$, which always assigns probability $\frac{1}{m}$ to all alternatives, is not $0$-randomly dictatorial if $g$ is not required to be strategyproof because it can be represented as $f_U=\frac{1}{m}d_i+\frac{m-1}{m}g$, where $d_i$ is the dictatorial SDS of voter $i$ and $g$ is the SDS that randomizes uniformly over all alternatives but voter $i$'s favorite one. Moreover, it should be mentioned that the maximality of $\gamma$ implies that $g$ is $0$-randomly dictatorial if $\gamma<1$. Otherwise, we could also represent $g$ as a mixture of a random dictatorship and some other strategyproof SDS $h$, which means that $f$ is $\gamma'$-randomly dictatorial for $\gamma'>\gamma$.

For a better understanding of $\gamma$-randomly dictatorial SDSs, we provide next a characterization of these SDSs. Recall for the following lemma that $R^{i:yx}$ denotes the profile derived from $R$ by only reinforcing $y$ against $x$ in voter $i$'s preference relation. 

\begin{restatable}{lemma}{zerodict}\label{lem:zerodict}
	A strategyproof SDS $f$ is $\gamma$-randomly dictatorial if and only if there are non-negative values $\gamma_1,\dots, \gamma_n$ such that: 
	\begin{enumerate}[label=\roman*), itemsep=0ex,parsep=1ex]
		\item $\sum_{i\in N}\gamma_i=\gamma$.
		\item $f(R^{i:yx},y)-f(R,y)\geq \gamma_i$ for all alternatives $x,y\in A$, voters $i\in N$, and preference profiles $R$ in which voter $i$ prefers $x$ the most and $y$ the second most.
		\item for every voter $i\in N$, there are alternatives $x,y\in A$ and a profile $R$ such that voter $i$ prefers $x$ the most and $y$ the second most in $R$, and $f(R^{i:yx},y)-f(R,y)= \gamma_i$.
	\end{enumerate} 
\end{restatable}

The proof of this lemma can be found in the appendix. \Cref{lem:zerodict} gives an intuitive interpretation of $\gamma$-randomly dictatorial SDSs: this axiom only requires that there are voters who always increase the winning probability of an alternative by at least $\gamma_i$ if they reinforce it to the first place. Hence, for small values of $\gamma$, this axiom is desirable as it only formulates a variant of strict monotonicity. However, for larger values of $\gamma$, $\gamma$-randomly dictatorial SDSs become more similar to random dictatorships. Furthermore, the proof of \Cref{lem:zerodict} shows that the decomposition of $\gamma$-randomly dictatorial SDSs is completely determined by the values $\gamma_1,\dots,\gamma_n$: given these values for an strategyproof SDS $f$, it can be represented as $f=\sum_{i\in N}\gamma_i d_i + (1-\sum_{i\in N} \gamma_i)g$, where $g$ is a strategyproof SDS and $d_i$ the dictatorial SDS of voter $i$.

\begin{figure}
	\[
	\begin{array}{c c c}
	1 & 1 & 1 \\
	\midrule
	a & b & c \\
	c & c & a \\
	b & a & b \\[1em]
	& R &
	\end{array}
	\qquad\qquad
	\begin{array}{c c c}
	1 & 1 & 1 \\
	\midrule
	a & b & c \\
	b & c & a \\
	c & a & b \\[1em]
	& R' &
	\end{array}
	\]
	\caption{Condorcet-consistent SDSs violate strategyproofness when $m=n=3$. Due to the symmetry of $R'$, we may assume without loss of generality that $f(R',a) > 0$. Since $f$ is Condorcet-consistent, it holds that $f(R,c) = 1$. Thus, voter 1 can manipulate by swapping $c$ and $b$ in $R$.}
	\label{fig:exampleProfile}
	\Description{Two Preference profiles R and R' with 3 voters and alternatives. Profile R' is the Condorcet cycle profile. The first voter has the preferences a, b, c, the second b, c, a, and the third c, a, b. Profile R is the same except that the first voter swaps his second and third ranked alternative so his new preferences are a, c, b. Alternative c is Condorcet winner in R.}
\end{figure}

Finally, we introduce $\alpha$-Condorcet-consistency. To this end, we first define the notion of a Condorcet winner. A \emph{Condorcet winner} is an alternative $x$ that wins every majority comparison according to preference profile $R$, i.e., $n_{xy}(R) > n_{yx}(R)$ for all $y \in A \setminus \{x\}$. \emph{Condorcet-consistency} demands that $f(R,x)=1$ for all preference profiles $R$ and alternatives $x$ such that $x$ is the Condorcet winner in $R$.
Unfortunately, Condorcet-consistency is in conflict with strategyproofness, which can easily be derived from Gibbard's random dictatorship theorem. 
A simple two-profile proof for this fact when $m=n=3$ is given in \Cref{fig:exampleProfile}. 
To circumvent this impossibility, we relax Condorcet-consistency: instead of requiring that the Condorcet winner always obtains probability $1$, we only require that it receives a probability of at least $\alpha$. This idea leads to \emph{$\alpha$-Condorcet-consistency}: an SDS $f$ satisfies this axiom if $f(R,x)\geq \alpha$ for all profiles $R$ and alternatives $x\in A$ such that $x$ is the Condorcet winner in $R$. For small values of $\alpha$, this axiom is clearly compatible with strategyproofness and therefore, we are interested in the maximum value of $\alpha$ such that there are $\alpha$-Condorcet-consistent and strategyproof SDSs. 

\begin{table}[t]
	\caption{Values of $\alpha$, $\beta$, and $\gamma$ for which specific SDSs are $\alpha$-Condorcet-consistent, $\beta$-\emph{ex post} efficient, and $\gamma$-randomly dictatorial. Each row shows the values of $\alpha$, $\beta$, and $\gamma$ for which a specific SDS satisfies the corresponding axioms. $f_\mathit{RD}$ abbreviates the uniform random dictatorship, $f_U$ the uniform lottery, $f_B$ the randomized Borda rule, and $f_C$ the randomized Copeland rule.}
	\label{tab:examples}
	\centering
	\renewcommand{\arraystretch}{1.4}
	\setlength{\tabcolsep}{1em}	
	\begin{tabular}{lccc}
		\toprule
		SDS & \makecell{$\alpha$-Condorcet\\-consistency} & \makecell{$\beta$-\emph{ex post}\\efficiency} & \makecell{$\gamma$-random\\dictatorship} \\ 
		\midrule
		$f_\mathit{RD}$ \hspace{2cm}& $0$ & $0$ & $1$ \\
		$f_\mathit{U}$ & $\frac{1}{m}$ & $\frac{1}{m}$ & $0$\\
		$f_\mathit{B}$ & $\frac{1}{m}+\frac{2 - (n \mod 2)}{mn}$ & $\frac{2(m-2)}{m(m-1)}$ & $\frac{2}{m(m-1)}$\\
		$f_\mathit{C}$ & $\frac{2}{m}$ & $\frac{2(m-2)}{m(m-1)}$ & 0\\
		\bottomrule
	\end{tabular}
	\renewcommand{\arraystretch}{1}
\end{table}

For a better understanding of $\alpha$-Condorcet-consistency, $\beta$-\emph{ex post} efficiency, and $\gamma$-random dictatorships, we discuss some of the values in \Cref{tab:examples} as examples. The uniform random dictatorship is $1$-randomly dictatorial and $0$-\emph{ex post} efficient by definition. Moreover, it is $0$-Condorcet-consistent because a Condorcet winner may not be top-ranked by any voter. The randomized Borda rule is $\frac{2(m-2)}{m(m-1)}$-\emph{ex post} efficient because it assigns this probability to an alternative that is second-ranked by every voter. Moreover, it is $\frac{2}{m(m-1)}$-randomly dictatorial as we can represent it as $\frac{2}{m(m-1)} f_\rd + \left(1-\frac{2}{m(m-1)}\right) g$, where $f_\rd$ is the uniform random dictatorship and $g$ is the point voting SDS defined by the scoring vector $\left(\frac{2(m-2)}{n(m(m-1) - 2)}, \frac{2(m-2)}{n(m(m-1) - 2)}, \frac{2(m-3)}{n(m(m-1) - 2)}, \dots, 0\right)$. Finally, the randomized Copeland rule is $0$-randomly dictatorial because there is for every voter a profile in which he can swap his two best alternatives without affecting the outcome. Moreover, it is $\frac{2}{m}$-Condorcet-consistent because a Condorcet winner $x$ satisfies that $n_{xy}(R)>\frac{n}{2}$ for all $y\in A\setminus \{x\}$ and hence, $f_C(R,x)=\sum_{y\in A\setminus\{x\}} b_{n_{xy}(R)}=(m-1)\frac{2}{m(m-1)}=\frac{2}{m}$. Note that \Cref{tab:examples} also contains a row corresponding to the uniform lottery. We consider this SDS as a threshold with respect to $\alpha$-Condorcet-consistency and $\beta$-\emph{ex post} efficiency because we can compute the uniform lottery without knowledge about the voters' preferences. Hence, if an SDS performs worse than the uniform lottery with respect to $\alpha$-Condorcet-consistency or $\beta$-\emph{ex post} efficiency, we could also dismiss the voters' preferences.

\section{Results}

In this section, we present our results about the $\alpha$-Condorcet-consistency and the $\beta$-\emph{ex post} efficiency of strategyproof SDSs.  
First, we prove that no strategyproof SDS satisfies $\alpha$-Condorcet-consistency for $\alpha>\frac{2}{m}$ and that the randomized Copeland rule $f_C$ is the only anonymous, neutral, and strategyproof SDS that satisfies $\alpha$-Condorcet-consistency for $\alpha= \frac{2}{m}$. Moreover, we show that every $\frac{1-\epsilon}{m}$-\emph{ex post} efficient and strategyproof SDS is $\gamma$-randomly dictatorial for $\gamma\geq \epsilon$. This statement can be seen as a continuous generalization of the random dictatorship theorem and implies, for instance, that every $0$-randomly dictatorial and strategyproof SDS can only satisfy $\beta$-\emph{ex post} efficiency for $\beta\geq\frac{1}{m}$, i.e., such SDSs are at least as inefficient as the uniform lottery. Even more, when additionally imposing anonymity and neutrality, we prove that only mixtures of the uniform random dictatorship and the uniform lottery satisfy this bound tightly, which shows that relaxing \emph{ex post} efficiency does not allow for appealing SDSs. In the last theorem, we identify a tradeoff between Condorcet-consistency and \emph{ex post} efficiency: no strategyproof SDS that satisfies $\alpha$-Condorcet consistency is $\beta$-\emph{ex post} efficient for $\beta < \frac{m - 2}{m - 1}\alpha$.
We derive these results through a series of lemmas. The proofs of all lemmas and \Cref{thm:contrd} are deferred to the appendix and we only present short proof sketches instead.

\subsection{$\alpha$-Condorcet-consistency}

As discussed in \Cref{subSec:Relax}, Condorcet-consistent SDSs violate strategyproofness. Therefore, we analyze the maximal $\alpha$ such that $\alpha$-Condorcet-consistency and strategyproofness are compatible. Our results show that strategyproofness only allows for a small degree of Condorcet-consistency: we prove that no strategyproof SDS satisfies $\alpha$-Condorcet-consistency for $\alpha>\frac{2}{m}$. This bound is tight as the randomized Copeland rule $f_C$ is $\frac{2}{m}$-Condorcet-consistent, which means that it is one of the ``most Condorcet-consistent'' strategyproof SDSs. Even more, we can turn this observation in a characterization of $f_C$ by additionally requiring anonymity and neutrality: the randomized Copeland rule is the only strategyproof SDS that satisfies $\frac{2}{m}$-Condorcet-consistency, anonymity, and neutrality.

For proving these results, we derive next a number of lemmas. As first step, we show in \Cref{lem:alphaAnonNeutral} that we can use a strategyproof and $\alpha$-Condorcet-consistent SDS to construct another strategyproof SDS that satisfies anonymity, neutrality, and $\alpha$-Condorcet-consistency for the same $\alpha$. 

\begin{restatable}{lemma}{alphaAnonNeutral}\label{lem:alphaAnonNeutral}
	If a strategyproof SDS satisfies $\alpha$-Condorcet-consistency for some $\alpha\in [0,1]$, there is also a strategyproof SDS that satisfies anonymity, neutrality, and $\alpha$-Condorcet-consistency for the same $\alpha$.
\end{restatable}

The central idea in the proof of \Cref{lem:alphaAnonNeutral} is the following: if there is a strategyproof and $\alpha$-Condorcet-consistent SDS $f$, then the SDS $f^{\pi\tau}(R,x)=f(\tau(\pi(R)), \tau(x))$ is also strategyproof and $\alpha$-Condorcet-consistent for all permutations $\pi:N\rightarrow N$ and $\tau:A\rightarrow A$. Since mixtures of strategyproof and $\alpha$-Condorcet-consistent SDSs are also strategyproof and $\alpha$-Condorcet-consistent, we can therefore construct an SDS that satisfies all requirements of the lemma by averaging over all permutations on $N$ and $A$. More formally, the SDS $f^*=\frac{1}{m!n!} \sum_{\pi\in \Pi} \sum_{\tau \in \Tau} f^{\pi\tau}$ (where $\Pi$ denotes the set of all permutations on $N$ and $\Tau$ the set of all permutations on $A$) meets all criteria of the lemma.

Due to \Cref{lem:alphaAnonNeutral}, we investigate next the $\alpha$-Condorcet-consistency of strategyproof SDSs that satisfy anonymity and neutrality. The reason for this is that this lemma turns an upper bound on $\alpha$ for these SDSs into an upper bound for all strategyproof SDSs. Since \Cref{thm:Barbera} shows that every strategyproof, anonymous, and neutral SDS can be decomposed in a point voting SDS and a supporting size SDS, we investigate these two classes separately in the following two lemmas. First, we bound the $\alpha$-Condorcet-consistency of point voting SDSs.

\begin{restatable}{lemma}{alphaPoint}\label{lem:alphaPoint}
	No point voting SDS is $\alpha$-Condorcet-consistent for $\alpha \ge \frac{2}{m}$ if $n \ge 3$ and $m \ge 3$.
\end{restatable}

The proof of this lemma relies on the observation that there can be $\lceil \frac{m}{2} \rceil$ Condorcet winner candidates, i.e., alternatives $x$ that can be made into the Condorcet winner by keeping $x$ at the same position in the preferences of every voter and only reordering the other alternatives. Since reordering the other alternatives does not affect the probability of $x$ in a point voting SDS, it follows that every Condorcet winner candidate has a probability of at least $\alpha$. Hence, we derive that $\alpha\leq \frac{1}{\lceil \frac{m}{2} \rceil} \le \frac{2}{m}$ and a slightly more involved argument shows that the inequality is strict. 

The last ingredient for the proof of \Cref{thm:Cond} is that no supporting size SDS can assign a probability of more than $\frac{2}{m}$ to any alternative. This immediately implies that no supporting size SDS satisfies $\alpha$-Condorcet-consistency for $\alpha>\frac{2}{m}$.

\begin{restatable}{lemma}{alphaSup}\label{lem:alphaSup}
	No supporting size SDS can assign more than $\frac{2}{m}$ probability to an alternative.
\end{restatable}

The proof of this lemma follows straightforwardly from the definition of supporting size SDSs. Each such SDS is defined by a scoring vector $(b_n, \dots, b_0)$ such that $b_i+b_{n-i}=\frac{2}{m(m-1)}$ for all $i\in \{0, \dots, n\}$ and $b_n\geq b_{n-1}\geq \dots \geq b_0\geq 0$. The probability of an alternative $x$ in a supporting size SDS $f$ is therefore bounded by $f(R, x) = \sum_{y\in A\setminus \{x\}} b_{n_{xy}(R)} \le (m-1)\frac{2}{m(m-1)} = \frac{2}{m}$. 

Finally, we have all necessary lemmas for the proof of our first theorem. 

\begin{theorem}\label{thm:Cond}
	The randomized Copeland rule is the only strategyproof SDS that satisfies anonymity, neutrality, and $\frac{2}{m}$-Condorcet-consistency if $m \ge 3$ and $n \ge 3$. Moreover, no strategyproof SDS satisfies $\alpha$-Condorcet-consistency for $\alpha > \frac{2}{m}$ if $n \ge 3$.
\end{theorem}

\begin{proof}
	The theorem consists of two claims: the characterization of the randomized Condorcet rule $f_C$ and the fact that no other strategyproof SDS can attain $\alpha$-Condorcet-consistency for a larger $\alpha$ than $f_C$. We prove these claims separately.\medskip 
	
	\textbf{Claim 1: The randomized Copeland rule is the only strategyproof SDS that satisfies $\frac{2}{m}$-Condorcet-consistency, anonymity, and neutrality if $m,n\geq 3$.}
	
	The randomized Copeland rule $f_C$ is a supporting size SDS and satisfies therefore anonymity, neutrality, and strategyproofness. Furthermore, it satisfies also $\frac{2}{m}$-Condorcet-consistency because a Condorcet winner $x$ wins every pairwise majority comparison in $R$. Hence, $n_{xy}(R)> \frac{n}{2}$ for all $y\in A\setminus \{x\}$, which implies that $f_C(R,x) = \sum_{y \in A \setminus \{x\}}b_{n_{xy}(R)} = (m - 1)\frac{2}{m(m-1)} = \frac{2}{m}$.
	
	Next, let $f$ be an SDS satisfying anonymity, neutrality, strategyproofness, and $\frac{2}{m}$-Condorcet-consistency. We show that $f$ is the randomized Copeland rule. Since $f$ is anonymous, neutral, and strategyproof, we can apply \Cref{thm:Barbera} to represent $f$ as $f=\lambda f_{\mathit{point}} + (1-\lambda) f_{\mathit{sup}}$, where $\lambda \in [0,1]$, $f_{\mathit{point}}$ is a point voting SDS, and $f_{\mathit{sup}}$ is a supporting size SDS. \Cref{lem:alphaPoint} states that there is a profile $R$ with Condorcet winner $x$ such that $f_{\mathit{point}}(R,x)<\frac{2}{m}$, and it follows from \Cref{lem:alphaSup} that $f_{\mathit{sup}}(R,x)\leq \frac{2}{m}$. Hence, $f(R,x)=\lambda f_{\mathit{point}}(R,x)+f_{\mathit{sup}}(R,x)<\frac{2}{m}$ if $\lambda >0$. Therefore, $f$ is a supporting size SDS as it satisfies $\frac{2}{m}$-Condorcet-consistency. 
	
	Next, we show that $f$ has the same scoring vector as the randomized Copeland rule. Since $f$ is a supporting size SDS, there is a scoring vector $b=(b_n, \dots, b_0)$ with $b_n \ge b_{n-1} \ge \dots \ge b_0 \ge 0$ and $b_i + b_{n - i} = \frac{2}{m(m-1)}$ for all $i\in \{1,\dots, n\}$ such that $f(R, x) = \sum_{y \in A\setminus \{x\} } b_{n_{xy}(R)}$. Moreover, $f(R, x) = \frac{2}{m}$ if $x$ is the Condorcet winner in $R$ because of $\frac{2}{m}$-Condorcet-consistency and \Cref{lem:alphaSup}. We derive from the definition of supporting size SDSs that the Condorcet winner $x$ can only achieve this probability if $b_{n_{xy(R)}}=\frac{2}{m(m-1)}$ for every other alternatives $y\in A\setminus \{x\}$. Moreover, observe that the Condorcet winner needs to win every majority comparison but is indifferent about the exact supporting sizes. Hence, it follows that $b_{i} = \frac{2}{m(m-1)}$ for all $i>\frac{n}{2}$ as otherwise, there is a profile in which the Condorcet winner does not receive a probability of $\frac{2}{m}$. We also know that $b_{i} + b_{n-i} = \frac{2}{m(m-1)}$, so $b_{i} = 0$ for all $i<\frac{n}{2}$. If $n$ is even, then $b_{\frac{n}{2}} = \frac{1}{m(m-1)}$ is required by the definition of supporting size SDSs as $\frac{n}{2}=n-\frac{n}{2}$. Hence, the scoring vector of $f$ is equivalent to the scoring vector of the randomized Copeland rule, which proves that $f$ is $f_C$.\medskip
	
	\textbf{Claim 2: No strategyproof SDS satisfies $\alpha$-Condorcet-consistency for $\alpha>\frac{2}{m}$ if $n\geq 3$.}
	
	The claim is trivially true if $m \leq 2$ because $\alpha$-Condorcet consistency for $\alpha>1$ is impossible. Hence, let $f$ denote a strategyproof SDS for $m\geq 3$ alternatives. We show in the sequel that $f$ cannot satisfy $\alpha$-Condorcet-consistency for $\alpha>\frac{2}{m}$. As a first step, we use \Cref{lem:alphaAnonNeutral} to construct a strategyproof SDS $f^*$ that satisfies anonymity, neutrality, and $\alpha$-Condorcet-consistency for the same $\alpha$ as $f$. Since $f^*$ is anonymous, neutral, and strategyproof, it follows from \Cref{thm:Barbera} that $f^*$ can be represented as a mixture of a point voting SDS $f_{\mathit{point}}$ and a supporting size SDS $f_{\mathit{sup}}$, i.e., $f^* = \lambda f_{\mathit{point}} + (1 - \lambda) f_{\mathit{sup}}$ for some $\lambda\in[0,1]$.
	
	Next, we consider $f_{\mathit{point}}$ and $f_{\mathit{sup}}$ separately. \Cref{lem:alphaPoint} implies for $f_{\mathit{point}}$ that there is a profile $R$ with a Condorcet winner $a$ such that $f_{\mathit{point}}(R,a)<\frac{2}{m}$. Moreover, \Cref{lem:alphaSup} shows that $f_{\mathit{sup}}(R,a)\leq \frac{2}{m}$ because supporting size SDSs never return a larger probability than $\frac{2}{m}$. Thus, we derive the following inequality, which shows that $f^*$ fails $\alpha$-Condorcet-consistency for $\alpha>\frac{2}{m}$. Hence, no strategyproof SDS satisfies $\alpha$-Condorcet-consistency for $\alpha > \frac{2}{m}$ when $n \ge 3$.
	\begin{equation*}
	\alpha \le f^*(R, a) =
	\lambda f_{\mathit{point}}(R, a) + (1 - \lambda) f_{\mathit{sup}}(R, a) \le
	\lambda \frac{2}{m} + (1 - \lambda)\frac{2}{m} =
	\frac{2}{m}
	\end{equation*}
	
\end{proof}

\begin{remark}
	 \Cref{lem:alphaAnonNeutral} can be applied to properties other than $\alpha$-Condorcet-consistency, too. For example, given a strategyproof and $\beta$-\emph{ex post} efficient SDS, we can construct another SDS that satisfies these axioms as well as anonymity and neutrality.
\end{remark}

\begin{remark}\label{rmk:independencecopeland}
	All axioms in the characterization of the randomized Copeland rule are independent of each other. The SDS that picks the Condorcet winner with probability $\frac{2}{m}$ if one exists and distributes the remaining probability uniformly between the other alternatives only violates strategyproofness. The randomized Borda rule satisfies all axioms of \Cref{thm:Cond} but $\frac{2}{m}$-Condorcet-consistency. An SDS that satisfies anonymity, strategyproofness, and $\frac{2}{m}$-Condorcet-consistency can be defined based on an arbitrary order of alternatives $x_0,\dots, x_{m-1}$. Then, we pick an index $i\in \{0, \dots, m-1\}$ uniformly at random and return the winner of the majority comparison between $x_i$ and $x_{i+1\bmod m}$ (if there is a majority tie, a fair coin toss decides the winner). 
	Finally, we can use the randomized Copeland rule $f_C$ to construct an SDS that fails only anonymity for even $n$: we just ignore one voter when computing the outcome of $f_C$. If $n$ is even and $x$ is the Condorcet winner in $R$, then $n_{xy}(R)\geq \frac{n+2}{2}$ for all $y\in N\setminus \{x\}$. Hence, the Condorcet winner remains a Condorcet winner after removing a single voter, which means that this SDS only fails anonymity. 
	
	Moreover, the impossibility in \Cref{thm:Cond} does not hold when there are only $n=2$ voters because random dictatorships are strategyproof and Condorcet-consistent in this case. The reason for this is that a Condorcet winner needs to be the most preferred alternative of both voters and is therefore chosen with probability $1$.
\end{remark}

\begin{remark}
	The randomized Copeland rule has multiple appealing interpretations. Firstly, it can be defined as a supporting size SDS as shown in \Cref{subsec:strategyproofness}. Alternatively, it can be defined as the SDS that picks two alternatives uniformly at random and then picks the majority winner between them; majority ties are broken by a fair coin toss. 
	Next, \Cref{thm:Cond} shows that the randomized Copeland rule is the SDS that maximizes the value of $\alpha$ for $\alpha$-Condorcet-consistency among all anonymous, neutral, and strategyproof SDSs.
	Finally, the randomized Copeland rule is the only strategyproof SDS that satisfies anonymity, neutrality, and assigns $0$ probability to a Condorcet loser whenever it exists. 
\end{remark}

\subsection{$\beta$-\emph{ex post} Efficiency}

According to Gibbard's random dictatorship theorem, random dictatorships are the only strategyproof SDSs that satisfy \emph{ex post} efficiency. In this section, we show that this result is rather robust by identifying a tradeoff between $\beta$-\emph{ex post} efficiency and $\gamma$-random dictatorships. More formally, we prove that for every $\epsilon\in [0,1]$, all strategyproof and $\frac{1-\epsilon}{m}$-\emph{ex post} efficient SDSs are $\gamma$-randomly dictatorial for $\gamma \ge \epsilon$. If we set $\epsilon=1$, we obtain the random dictatorship theorem. On the other hand, we derive from this theorem that every $0$-randomly dictatorial and strategyproof SDS is $\beta$-\emph{ex post} efficient for $\beta\geq \frac{1}{m}$, i.e., every such SDS is at least as inefficient as the uniform lottery. Moreover, we prove for every $\epsilon\in [0,1]$ that mixtures of the uniform random dictatorship and the uniform lottery are the only $\epsilon$-randomly dictatorial SDSs that satisfy anonymity, neutrality, strategyproofness, and $\frac{1-\epsilon}{m}$-\emph{ex post} efficiency. In summary, these results demonstrate that relaxing \emph{ex post} efficiency does not lead to particularly appealing strategyproof SDSs. 
Furthermore, we also identify a tradeoff between $\alpha$-Condorcet-consistency and $\beta$-\emph{ex post} efficiency: every $\alpha$-Condorcet consistent and strategyproof SDS fails $\beta$-\emph{ex post} efficiency for $\beta < \frac{m - 1}{m - 2}\alpha$. Under the additional assumption of anonymity and neutrality, we characterize the strategyproof SDSs that maximize the ratio between $\alpha$ and $\beta$: all these SDSs are mixtures of the randomized Copeland rule and the uniform random dictatorship.

For proving the tradeoff between $\beta$-\emph{ex post} efficiency and $\gamma$-random dictatorships, we first investigate the efficiency of $0$-randomly dictatorial strategyproof SDSs. In more detail, we prove next that every such SDS fails $\beta$-\emph{ex post} efficiency for $\beta< \frac{1}{m}$.

\begin{restatable}{lemma}{expost}\label{lem:expost}
		No strategyproof SDS that is $0$-randomly dictatorial satisfies $\beta$-\emph{ex post} efficiency for $\beta < \frac{1}{m}$ if $m \ge 3$.
\end{restatable}
The proof of this result is quite similar to the one for the upper bound on $\alpha$-Condorcet-consistency in \Cref{thm:Cond}. In particular, we first show that all $0$-randomly mixtures of duples and all $0$-randomly dictatorial mixtures of unilaterals violate $\beta$-\emph{ex post} efficiency for $\beta< \frac{1}{m}$. Next, we consider an arbitrary $0$-randomly dictatorial SDS $f$ and aim to show that there are a profile $R$ and a Pareto-dominated alternative $x\in A$ such that $f(R,x)\geq \beta$. Even though \Cref{thm:Gibbard1} allows us to represent $f$ as the convex combination of a $0$-randomly dictatorial mixture of unilaterals $f_\mathit{uni}$ and a mixture of duples $f_\mathit{duple}$, our previous observations have unfortunately no direct consequences for the $\beta$-\emph{ex post} efficiency of $f$. The reason for this is that $f_\mathit{uni}$ and $f_\mathit{duple}$ might violate $\beta$-\emph{ex post} efficiency for different profiles or alternatives. We solve this problem by transforming $f$ into a $0$-randomly dictatorial SDS $f^*$ that is $\beta$-\emph{ex post} efficient for the same $\beta$ as $f$ and satisfies additional properties. In particular, $f^*$ can be represented as a convex combination of a $0$-randomly dictatorial mixture of unilaterals $f_\mathit{uni}^*$ and a $0$-randomly dictatorial mixture of duples $f_\mathit{duple}^*$ such that $f_\mathit{uni}^*(R,x)\geq \frac{1}{m}$ and $f_\mathit{duple}^*(R,x)\geq \frac{1}{m}$ for some profile $R$ in which alternative $x$ is Pareto-dominated. Consequently, $f^*$ fails $\beta$-\emph{ex post} efficiency for $\beta<\frac{1}{m}$, which implies that also $f$ violates this axiom. 

Based on \Cref{lem:expost}, we can now show the tradeoff between \emph{ex post} efficiency and the similarity to a random dictatorship.

\begin{restatable}{theorem}{contrd}\label{thm:contrd}
	For every $\epsilon\in [0,1]$, every strategyproof and $\frac{1 - \epsilon}{m}$-\emph{ex post} efficient SDS is $\gamma$-randomly dictatorial for $\gamma \ge \epsilon$ if $m\geq 3$. Moreover, if $\gamma = \epsilon$, $m \ge 4$, and the SDS satisfies additionally anonymity and neutrality, it is a mixture of the uniform random dictatorship and the uniform lottery.
\end{restatable}

The proof of the first claim follows easily from \Cref{lem:expost}: we consider a strategyproof SDS $f$ and use the definition of $\gamma$-randomly dictatorial SDSs to represent $f$ as a mixture of a random dictatorship and another strategyproof SDS $g$. Unless $f$ is a random dictatorship, the maximality of $\gamma$ entails that $g$ is $0$-randomly dictatorial. Hence, \Cref{lem:expost} implies that $g$ can only be $\beta$-\emph{ex post} efficient for $\beta \geq \frac{1}{m}$. Consequently, $\gamma\geq \epsilon$ must be true if $f$ satisfies $\frac{1-\epsilon}{m}$-\emph{ex post} efficiency. For the second claim, we observe first that every anonymous, neutral, and strategyproof SDS $f$ can be represented as a mixture of the uniform random dictatorship and another strategyproof, anonymous, and neutral SDS $g$. Moreover, unless $f$ is $1$-randomly dictatorial, $g$ is $0$-randomly dictatorial. Thus, \Cref{lem:expost} and the assumption that $\gamma=\epsilon$ require that $g$ is exactly $\frac{1}{m}$-\emph{ex post} efficient. Finally, the claim follows by proving that the uniform lottery is the only $0$-randomly dictatorial and strategyproof SDS that satisfies anonymity, neutrality, and $\frac{1}{m}$-\emph{ex post} efficiency if $m \ge 4$. For $m = 3$ the randomized Copeland rule also satisfies all required axioms and the uniform rule is thus not the unique choice.

	\Cref{thm:contrd} represents a continuous strengthening of Gibbard's random dictatorship theorem: the more \emph{ex post} efficiency is required, the closer a strategyproof SDS gets to a random dictatorship. Conversely, our result also entails that $\gamma$-randomly dictatorial SDSs can only satisfy $\frac{1 - \epsilon}{m}$-\emph{ex post} efficiency for $\epsilon \le \gamma$. Moreover, the second part of the theorem indicates that relaxing \emph{ex post} efficiency does not allow for particularly appealing strategyproof SDSs.
	
The correlation between $\beta$-\emph{ex post} efficiency and $\gamma$-randomly dictatorships also suggests a tradeoff between $\alpha$-Condorcet-consistency and $\beta$-\emph{ex post} efficiency because all random dictatorships are $0$-Condorcet-consistent for sufficiently large $m$ and $n$. Perhaps surprisingly, we show next that $\alpha$-Condorcet consistency and $\beta$-\emph{ex post} efficiency are in relation with each other for strategyproof SDSs. As a consequence of this insight, two strategyproof SDSs are particularly interesting: random dictatorships because they are the most \emph{ex post} efficient SDSs, and the randomized Copeland rule because it is the most Condorcet-consistent SDS.

\begin{restatable}{theorem}{alphabeta}\label{thm:alphabeta}
	Every strategyproof SDS that satisfies anonymity, neutrality, $\alpha$-Condorcet consistency, and $\beta$-\emph{ex post} efficiency with $\beta = \frac{m-2}{m-1}\alpha$ is a mixture of the uniform random dictatorship and the randomized Copeland rule if $m\geq 4$, $n\geq 5$. Furthermore, there is no strategyproof SDS with $\beta < \frac{m-2}{m-1}\alpha$ if $m\geq 4$, $n\geq 5$.
\end{restatable}

\begin{proof}
	Let $f$ be a strategyproof SDS that satisfies $\alpha$-Condorcet consistency for some $\alpha\in [0,\frac{2}{m}]$ and let $\beta\in [0,1]$ denote the minimal value such that $f$ is $\beta$-\emph{ex post} efficient. We first show that $\beta\geq \frac{m-2}{m-1}\alpha$ and hence apply \Cref{lem:alphaAnonNeutral} to construct an SDS $f'$ that satisfies strategyproofness, anonymity, neutrality, $\alpha'$-Condorcet consistency for $\alpha' \ge \alpha$, and $\beta'$-\emph{ex post} efficiency for $\beta' \le \beta$. In particular, if $f'$ is only $\beta'$-\emph{ex post} efficient for $\beta'\geq \frac{m-2}{m-1}\alpha'$, then $f$ can only satisfy $\beta$-\emph{ex post} efficiency for $\beta\geq \beta'\geq \frac{m-2}{m-1}\alpha'\geq  \frac{m-2}{m-1}\alpha$.
	
	Since $f'$ satisfies anonymity, neutrality, and strategyproofness, we can apply \Cref{thm:Barbera} to represent it as a mixture of a supporting size SDS and a point voting SDS, i.e., $f' = \lambda f_\mathit{point} + (1 - \lambda) f_\mathit{sup}$ for some $\lambda\in[0,1]$. Let $(a_1, \dots, a_m)$ and $(b_0, \dots, b_n)$ denote the scoring vectors describing $f_\mathit{point}$ and $f_\mathit{sup}$, respectively. Next, we a derive lower bound for $\alpha'$ and an upper bound for $\beta'$ by considering specific profiles. First, consider the profile $R$ in which every voter reports $a$ as his best alternative and $b$ as his second best alternative; the remaining alternatives can be ordered arbitrarily. It follows from the definition of point voting SDSs that $f_\mathit{point}(R,b)=na_2$ and from the definition of supporting size SDS that $f_\mathit{sup}(R,b)=(m-2)b_n+b_0$. Since $a$ Pareto-dominates $b$ in $R$, it follows that $\beta'\geq f(R,b)=\lambda na_2 + (1-\lambda)((m-2)b_n+ b_0)$. 
	
	For the upper bound on $\alpha$, consider the following profile $R'$ where alternative $x$ is never ranked first, but it is the Condorcet winner and wins every pairwise comparison only with minimal margin. We denote for the definition of $R'$ the alternatives as $A = \{ x, x_1, \dots, x_{m - 1} \}$. In $R'$, the voters $i \in \{1, 2, 3\}$ ranks alternatives $X_i := \{ x_k \in A \setminus \{x\}: k \mod 3 = i-1 \}$ above $x$ and all other alternatives below. Since $m \ge 4$, none of them ranks $x$ first. If the number of voters $n$ is even, we duplicate voters $1$, $2$, and $3$. As last step, we add pairs of voters with inverse preferences such that no voter prefers $x$ the most until $R'$ consists of $n$ voters. Since alternative $x$ is never top-ranked in $R'$, it follows that $f_\mathit{point}(R',x) \le na_2$. Furthermore, $n_{xy}(R')=\lceil\frac{n + 1}{2}\rceil$ for all $y\in A\setminus \{x\}$ and therefore $f_\mathit{sup}(R',x) = (m-1) b_{\lceil\frac{n + 1}{2}\rceil}$. Finally, we derive that $\alpha'\leq f(R',x)\leq\lambda na_2+(1-\lambda)(m-1)b_{\lceil\frac{n + 1}{2}\rceil}$ because $x$ is by construction the Condorcet winner in $R'$.
	
	Using these bounds, we show next that $f'$ is only $\beta'$-\emph{ex post} efficiency for $\beta'\geq \frac{m-2}{m-1}\alpha'$, which proves the second claim of the theorem. In the subsequent calculation, the first and last inequality follow from our previous analysis. The second inequality is true since $\frac{m-2}{m-1}\leq 1$ and $\frac{m-2}{m-1}(m-1)=(m-2)$. The third inequality uses the definition of supporting size SDSs.
	\begin{align*}
		\beta'&\geq \lambda na_2 + (1-\lambda)((m-2)b_n+ b_0)\\
		& \geq \frac{m-2}{m-1}\lambda na_2 + \frac{m-2}{m-1}(1-\lambda)((m-1)b_n + b_0)\\
		& \geq \frac{m-2}{m-1}\lambda na_2 + \frac{m-2}{m-1}(1-\lambda)(m-1)b_{\lceil\frac{n + 1}{2}\rceil}\\
		& \geq \frac{m-2}{m-1}\alpha'
	\end{align*}
		
	Finally, note that, if $\beta' = \frac{m-2}{m-1}\alpha'$, all inequalities must be tight. If the second inequality is tight $a_2 = 0$ and $b_0 = 0$, and when the third inequality is tight $b_n = b_{\lceil\frac{n + 1}{2}\rceil}$. These observations fully specify the scoring vectors of $f_\mathit{point}$ and $f_\mathit{sup}$. For the point voting SDS, $a_2 = 0$ implies $a_i = 0$ for all $i\geq 2$ and $a_1 = \frac{1}{n}$, i.e., $f_\mathit{point}$ is the uniform random dictatorship. Next, $b_0 = 0$ and $b_n = b_{\lceil\frac{n + 1}{2}\rceil}$ imply that $b_i =\frac{2}{m(m-1)}$ for all $i\in \{\lceil\frac{n + 1}{2}\rceil,\dots, b_n\}$ and $b_i = 0$ for all $i \in \{0, \dots, \lfloor\frac{n-1}{2}\rfloor \}$. Moreover, if $n$ is even, the definition of supporting size SDSs requires that $b_{\frac{n}{2}} = \frac{1}{m(m-1)}$. This shows that $f_\mathit{sup}$ is the randomized Copeland rule. Consequently, the SDS $f'$ is a mixture of the uniform random dictatorship and the randomized Copeland rule if $\beta' = \frac{m-2}{m-1}\alpha'$. This proves that every strategyproof SDS that satisfies anonymity, neutrality, $\alpha$-Condorcet consistency, and $\beta$-\emph{ex post} efficiency with $\beta = \frac{m-2}{m-1}\alpha$ is a mixture of the uniform random dictatorship and the randomized Copeland rule.
\end{proof}

\begin{remark}
	All axioms of the characterization in \Cref{thm:alphabeta} are independent of each other. Every mixture of random dictatorships other than the uniform one and the randomized Copeland rule only violates anonymity. An SDS that violates only neutrality can be constructed by using a variant of the randomized Copeland rule that does not split the probability equally if there is a majority tie.
	Finally, the correlation between $\alpha$-Condorcet-consistency and $\beta$-\emph{ex post} efficiency is required since the uniform lottery satisfies all other axioms. 
	Moreover, all bounds on $m$ and $n$ in \Cref{thm:alphabeta} are tight. If there are only $n=2$ voters, $m=3$ alternatives, or $m=4$ alternatives and $n=4$ voters, the uniform random dictatorship is not $0$-Condorcet consistent since a Condorcet winner is always ranked first by at least one voter. Hence, the bound on $\beta$ does not hold in these cases. In contrast, our proof shows that \Cref{thm:alphabeta} is also true when $n = 3$.
\end{remark}

\section{Conclusion}

\begin{figure}[t]
	\begin{tikzpicture}
	    \fill[color=gray!20] (0,3) -- (3,3) -- (3,2) -- (5/3, 2) -- (0,0);

		\draw[->] (0,0) -- (3,0) node[right] {$\beta$};
		\draw[->] (0,0) -- (0,3) node[above] {$\alpha$};
		\node[below] (Z) at (0, 0) {$0$};
		\draw[shift={(1,0)}] (0pt,2pt) -- (0pt,-2pt) node[, below] {$\frac{\vphantom{(}1}{\vphantom{(}m}$};
		\draw[shift={(1,0)}, densely dotted] (0,0) -- (0,3);
		\draw[shift={(1.66,0)}] (0pt,2pt) -- (0pt,-2pt) node[shift={(0pt, 0)}, below] {$\frac{2(m - 2)}{m(m - 1)}$};
		\draw[shift={(1.66,0)}, densely dotted] (0,0) -- (0,3);
		\draw[shift={(0,0)}] (2pt,0pt) -- (-2pt,0pt);
		\draw[shift={(0,0)}] (0pt,2pt) -- (0pt,-2pt);
		\draw[shift={(0,1)}] (2pt,0pt) -- (-2pt,0pt) node[left] {$\frac{1}{m}$};
		\draw[shift={(0,1)}, densely dotted] (0,0) -- (3, 0);
		\draw[shift={(0,2)}] (2pt,0pt) -- (-2pt,0pt) node[left] {$\frac{2}{m}$};
		\draw[shift={(0,2)}, densely dotted] (0,0) -- (3, 0);
		\draw[shift={(0,0)}, densely dotted] (0,0) -- (1.66, 2);
		
		\draw[shift={(0, 0)}, line width=0.6pt] (-2pt,-2pt) -- (2pt,2pt) (-2pt,2pt) -- (2pt,-2pt) node[above right] {$d$};
		\draw[shift={(1.66, 2)}, line width=0.6pt] (-2pt,-2pt) -- (2pt,2pt) (-2pt,2pt) -- (2pt,-2pt) node[right] {$c$};
		\draw[shift={(1.66, 1.2)}, line width=0.6pt] (-2pt,-2pt) -- (2pt,2pt) (-2pt,2pt) -- (2pt,-2pt) node[right] {$b$};
		\draw[shift={(1, 1)}, line width=0.6pt] (-2pt,-2pt) -- (2pt,2pt) (-2pt,2pt) -- (2pt,-2pt) node[right] {$u$};
		
	\end{tikzpicture}
	\quad
	\begin{tikzpicture}
	    \fill[color=gray!20] (0,0) -- (0,3) -- (3/4,0) -- (0,0);
	    \fill[color=gray!20] (3,3) -- (0,3) -- (3,0) -- (3,3);
		
		\draw[-] (0,0) -- (3,0) node[right] {$\beta$};
		\draw[-] (0,0) -- (0,3) node[above] {$\gamma$};
		\node[below] (Z) at (0, 0) {$0$};
		\draw[shift={(3/4,0)}] (0pt,2pt) -- (0pt,-2pt) node[, below] {$\frac{\vphantom{(}1}{\vphantom{(}m}$};
		\draw[shift={(1,0)}] (0pt,2pt) -- (0pt,-2pt) node[shift={(1.5em, 0)}, below] {$\frac{2(m - 2)}{m(m - 1)}$};
		\draw[shift={(1,0)}, densely dotted] (0,0) -- (0,3);
		\draw[shift={(0,0)}] (2pt,0pt) -- (-2pt,0pt);
		\draw[shift={(0,0)}] (0pt,2pt) -- (0pt,-2pt);
		\draw[shift={(3,0)}] (0pt,2pt) -- (0pt,-2pt) node[shift={(0pt, 0)}, below] {$1$};
		\draw[shift={(0,3)}] (2pt,0pt) -- (-2pt,0pt) node[left] {$1$};
		\draw[shift={(0,0)}, densely dotted] (0,3) -- (3/4, 0);
		
		\draw[shift={(0,0)}, densely dotted] (0,3) -- (3, 0);

		\draw[shift={(0, 3)}, line width=0.6pt] (-2pt,-2pt) -- (2pt,2pt) (-2pt,2pt) -- (2pt,-2pt) node[above right] {$d$};
		\draw[shift={(1, 0)}, line width=0.6pt] (-2pt,-2pt) -- (2pt,2pt) (-2pt,2pt) -- (2pt,-2pt) node[above right] {$c$};
		\draw[shift={(1, 2.75/10)}, line width=0.6pt] (-2pt,-2pt) -- (2pt,2pt) (-2pt,2pt) -- (2pt,-2pt) node[above right] {$b$};
		\draw[shift={(3/4, 0)}, line width=0.6pt] (-2pt,-2pt) -- (2pt,2pt) (-2pt,2pt) -- (2pt,-2pt) node[above left] {$u$};
		
	\end{tikzpicture}
	\caption{Graphical summary of our results. Points in the figures correspond to SDSs and the horizontal axis indicates in both figures the value of $\beta$ for which the considered SDS is $\beta$-\emph{ex post} efficient. In the left figure, the vertical axis states the $\alpha$ for which the considered SDSs are $\alpha$-Condorcet-consistent, and in the right figure, it shows the $\gamma$ for which SDSs are $\gamma$-randomly dictatorial. \Cref{thm:Cond,thm:alphabeta} show that no strategyproof SDS lies in the grey area of the left figure. \Cref{thm:contrd} shows that no strategyproof SDS lies in the grey area below the diagonal in the right figure. Furthermore, no SDS lies in the grey area above the diagonal since a $\gamma$-randomly dictatorial SDS can put no more than $1 - \gamma$ probability on Pareto-dominated alternatives. Finally, the following SDS are marked in the figures: $d$ corresponds to all random dictatorships, $c$ to the randomized Copeland rule, $b$ to the randomized Borda rule, and $u$ to the uniform lottery.}
	\label{fig:resultGraph}
\end{figure}
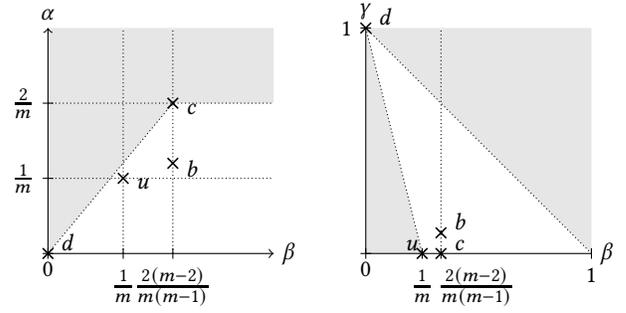

In this paper, we analyzed strategyproof SDSs by considering relaxations of Condorcet-consistency and \emph{ex post} efficiency. Our findings, which are summarized in \Cref{fig:resultGraph}, show that two strategyproof SDSs perform particularly well with respect to these axioms: the uniform random dictatorship (and random dictatorships in general), and the randomized Copeland rule. In more detail, we prove that the randomized Copeland rule is the only strategyproof, anonymous, and neutral SDS which guarantees a probability of $\frac{2}{m}$ to the Condorcet winner. Since no other strategyproof SDS can guarantee more probability to the Condorcet winner (even if we drop anonymity and neutrality), this characterization identifies the randomized Copeland rule as one of the most Condorcet-consistent strategyproof SDSs. On the other hand, Gibbard's random dictatorship theorem shows that random dictatorships are the only \emph{ex post} efficient and strategyproof SDSs. We present a continuous generalization of this result: for every $\epsilon\in[0,1]$, every $\frac{1-\epsilon}{m}$-\emph{ex post} efficient and strategyproof SDS is $\gamma$-randomly dictatorial for $\gamma\geq \epsilon$. This means informally that, even if we allow that Pareto-dominated alternatives can get a small amount of probability, we end up with an SDS similar to a random dictatorship. Finally, we derive a tradeoff between $\alpha$-Condorcet-consistency and $\beta$-\emph{ex post} efficiency for strategyproof SDSs: every strategyproof and $\alpha$-Condorcet-consistent SDS fails $\beta$-\emph{ex post} efficiency for $\beta<\frac{m-2}{m-1} \alpha$. This theorem entails that it is not possible to jointly optimize these two axioms, which highlights the special role of the randomized Copeland rule and random dictatorships again.

\balance

\begin{acks}
This work was supported by the Deutsche Forschungsgemeinschaft under grant \mbox{BR 2312/12-1}. We thank Dominik Peters for stimulating discussions and the anonymous reviewers for their helpful comments.
\end{acks}

\bibliographystyle{ACM-Reference-Format}

\newpage

\section*{Appendix: Omitted Proofs}

Here, we discuss the missing proofs of all lemmas and of \Cref{thm:contrd}. Proof sketches providing intuition for the lemmas can be found in the main body. First, we discuss the proof of \Cref{lem:zerodict}. Recall for this proof that $R^{i:yx}$ is the profile derived from $R$ by letting voter $i$ reinforce $y$ against $x$.

\zerodict*
\begin{proof}
	``$\impliedby$'' Assume that $f$ is a strategyproof SDS for which there are values $\gamma_1, \dots, \gamma_n$ such that $f(R^{i:yx},y)-f(R,y)\geq \gamma_i\geq 0$ for all alternatives $x,y\in A$, voters $i\in N$, and profiles $R$ such that voter $i$ prefers $x$ the most and $y$ the second most in $R$. Furthermore, we assume that for every voter $i\in N$, this inequality is tight for at least one pair of alternatives $x,y\in A$ and one such profile $R$. We show next that $f$ is $\gamma$-randomly dictatorial for $\gamma=\sum_{i\in N} \gamma_i$. 
	
	As first step, note that $f(R,x)\geq \sum_{i\in S} \gamma_i$ for every profile $R$, alternative $x\in A$, and set of voters $S\subseteq N$ such that all voters in $S$ report $x$ as their favorite alternative. This follows by letting the voters $i\in S$ one after another swap $x$ with their second best alternative $y$ (note that $y$ might be a different alternative for every voter $i\in S$). Using our assumption on $f$, the probability of $y$ has to increase by at least $\gamma_i$ during such a step, which means that the probability of $x$ decreases by $\gamma_i$ because of localizedness. Furthermore, it holds that $f(R',x)\geq 0$, where $R'$ is the profile derived by letting all voters in $S$ swap their best two alternatives. Combining these two facts then implies that $f(R,x)\geq \sum_{i\in S} \gamma_i$. Note that this observation implies that $\gamma\leq 1$ because otherwise, $f$ cannot be a valid SDS. Moreover, $f$ is a random dictatorship if $\gamma=1$. This follows from the following reasoning: for all profiles $R$ and alternatives $x\in A$, it holds that $f(R,x)\geq \sum_{i\in S_x} \gamma_i$, where $S_x$ denotes the set of voters who prefer $x$ the most in $R$. Since the sets $S_x$ partition $N$ and $\gamma=1$, this inequality must be tight for every alternative; otherwise, $\sum_{x\in A} f(R,x)>\sum_{x\in A} \sum_{i\in S_x} \gamma_i=1$, contradicting the definition of an SDS. Hence, if $\gamma=1$, $f$ is $1$-randomly dictatorial as $f=\sum_{i\in N}\gamma_i d_i$, where $d_i$ denotes the dictatorial SDS of voter $i$.
	
	As next case, suppose that $\gamma<1$ and define $g=\frac{1}{1 - \gamma}\left(f -\sum_{i\in N}\gamma_i d_i\right)$. Note that $g$ is a well-defined SDS: for all profiles $R$ and alternatives $g$, it holds that $g(R,x)\geq 0$ because $f(R,x)\geq \sum_{i\in S_x} \gamma_i$. Moreover, $\sum_{x\in A}g(R)=\frac{1}{1 - \gamma} \sum_{x\in A}f(R,x) -\sum_{x\in A}\sum_{i\in N}\frac{\gamma_i}{1 - \gamma} d_i=\frac{1}{1 - \gamma}-\frac{\gamma}{1 - \gamma}=1$ for all profiles $R$. Next, we show that $g$ is strategyproof, which implies that $f$ is $\gamma'$-randomly dictatorial for $\gamma'\geq\gamma$ because $f=\sum_{i\in N}\gamma_i d_i + (1-\gamma) g$. It is sufficient to show that $g$ is localized and non-perverse because then \Cref{thm:Gibbard2} implies that $g$ is strategyproof. In more detail, $g$ is localized because the SDS $f$ and all SDSs $d_i$ are localized. Hence, swapping two alternatives in the preferences of a voter only affects these two alternatives. For seeing that $g$ is non-perverse, consider a voter $i$, two alternatives $x,y\in A$ and a profile $R$ such that $x$ is voter $i$'s $k$-th best alternative and $y$ is his $k+1$-th best one. We show that $g(R^{i:yx}, y)\geq g(R,y)$, which entails that $g$ is non-perverse. Note for this that $d_j(R^{i:yx})=d_j(R)$ for all $j\in N\setminus \{i\}$ because the preferences of these voters did not change, and $f(R^{i:yx},y)- f(R,y)\geq 0$ because $f$ is strategyproof. If $x$ and $y$ are not the two best alternatives of voter $i$, then $d_i(R^{i:yx})=d_i(R) = 0$. Hence, it immediately follows that $g(R^{i:yx},y)-g(R,y)=\frac{1}{1-\gamma}\Big(f(R^{i:yx},y)-f(R,y)\Big)\geq 0$ in this case. On the other hand, if $x$ and $y$ are voter $i$'s two best alternative, we have that $d_i(R^{i:yx},y)=1$ and $d_i(R,y)=0$. Moreover, our assumptions imply that $f(R^{i:yx},y)-f(R,y)\geq \gamma_i$ because $x$ and $y$ voter $i$'s two best alternatives. Thus, we calculate that $g(R^{i:yx},y)-g(R,y)=\frac{1}{1-\gamma}\Big(f(R^{i:yx},y)-f(R,y) - \gamma_i (d_i(R^{i:yx},y) - d_i(R,y))\Big)\geq \frac{1}{1-\gamma}\Big(\gamma_i -\gamma_i\Big)=0$, which shows that $g$ is non-perverse.
	
	Finally, we show that $f$ cannot be $\gamma'$-randomly dictatorial for $\gamma'>\gamma$. If this was the case, we can represent $f$ as $f=\sum_{i\in N}\gamma'_i d_i + (1-\gamma') g'$, where $\gamma'_i\geq 0$ are values such that $\sum_{i\in N}\gamma'_i=\gamma'$ and $g'$ is a strategyproof SDS. Since $\gamma'>\gamma$, there is a voter $i$ with $\gamma_i'>\gamma_i$. Furthermore, our assumptions state that there are a profile $R$ and alternatives $x,y$ such that voter $i$ prefers $x$ the most and $y$ the second most in $R$, and $f(R^{i:yx},y)-f(R,y)=\gamma_i$. This means that $\Big(f(R^{i:yx},y)-\sum_{j\in N}\gamma'_j d_j(R^{i:yx},y)\Big)-\Big(f(R,y)-\sum_{j\in N}\gamma'_j d_j(R,y)\Big)=\gamma_i-\gamma_i'<0$ because $d_i(R^{i:yx},y)-d_i(R,y)=1$ and $d_j(R^{i:yx},y)-d_j(R,y)=0$ for all $j\in N\setminus \{i\}$. Consequently, $g'(R^{i:yx},y)-g'(R,y)<0$ which means that $g'$ violates non-perversity and therefore also strategyproofness. Hence, the assumption that $f$ is $\gamma'$-randomly dictatorial for $\gamma'>\gamma$ is wrong and $f$ is therefore $\gamma$-randomly dictatorial.
	
	``$\implies$'' Let $f$ be a strategyproof $\gamma$-randomly dictatorial SDS. We show next that there are values $\gamma_i$ that satisfy the requirements of the lemma. Since $f$ is $\gamma$-randomly dictatorial, it can be represented as $f=\gamma d+(1-\gamma) g$, where $d$ is a random dictatorship and $g$ is another strategyproof SDS. Moreover, as $d$ is a random dictatorship, there are values $\delta_1,\dots,\delta_n$ such that $\delta_i\geq0$ for all $i\in N$, $\sum_{i\in N} \delta_i=1$, and $d=\sum_{i\in N} \delta_i d_i$. In the last equation, $d_i$ denotes the dictatorial SDS of voter $i$. Combining these two equations, we derive that $f=\gamma \sum_{i\in N} \delta_i d_i + (1-\gamma) g$. We show in the sequel that the values $\gamma_i=\gamma\delta_i$ satisfy all requirements of our lemma. First, note that the conditions $\gamma_i\geq 0$ for all $i\in N$ and $\sum_{i\in N}\gamma_i=\gamma$ are obviously true.
	
	Next, consider two alternatives $x,y\in A$, an arbitrary voter $i\in N$, and a profile $R$ in which voter $i$ reports $x$ as his best alternative and $y$ as his second best one. It holds that $g(R^{i:yx},y)-g(R,y)\geq 0$ because $g$ is strategyproof and therefore non-perverse, $d_j(R^{i:yx},y)-d_j(R,y)=0$ for all $j\in N\setminus \{i\}$ because $\succeq^{i:yx}_j\ =\ \succeq_j$, and $d_i(R^{i:yx},y)-d_i(R,y)=1$ as $y$ is voter $i$'s best alternative in $R^{i:yx}$, but not in $R$. Hence, it follows that $f(R^{i:yx},y)-f(R,y)\geq \gamma\delta_i = \gamma_i$ for all voters $i\in N$, alternatives $x,y\in A$, and preference profiles $R$ in which voter $i$ reports $x$ as his best and $y$ as his second best alternative. 
	
	Finally, it remains to show that there is for every voter $i\in N$ a pair of alternatives $x,y\in A$ and a profile $R$ such that voter $i$ prefers $x$ the most and $y$ the second most in $R$ and $f(R^{i:yx},y)-f(R,y)=\gamma_i$. Assume this is not the case for some voter $i$, i.e, that $f(R^{i:yx},y)-f(R,y)>\gamma_i$ for all alternatives $x,y\in A$ and profiles $R$ in which $x$ is voter $i$'s best alternative and $y$ his second best one. Hence, let $\gamma_i'>\gamma_i$ denote the minimal value of $f(R^{i:yx},y)-f(R,y)$ among all alternatives $x,y\in A$ and preference profiles $R$ in which voter $i$ reports $x$ as his best alternative and $y$ as his second best one. Moreover, define $\gamma'=\gamma_i+ \sum_{j\in N\setminus \{i\}}\gamma_j$. We can now apply the arguments for the inverse direction to derive that $f$ is $\gamma''$-randomly dictatorial for some $\gamma''\geq \gamma'>\gamma$. This contradicts our assumption that $f$ is $\gamma$-randomly dictatorial as $\gamma$ must be the maximal value such that $f$ can be represented as $f=\gamma d + (1-\gamma) g$, where $d$ is a random dictatorship and $g$ is another strategyproof SDS. Hence, it follows that there are for every voter $i\in N$ a profile $R$ and two alternatives $x,y\in A$ such that $f(R^{i:yx},y)-f(R,y)=\gamma_i$ and voter $i$ reports $x$ as his best alternative and $y$ as his second best one in $R$. This means that our choice of $\gamma_i$ satisfies all requirements of the lemma. 
\end{proof}

\subsection{Proof of \Cref{thm:Cond}}

Next, we show the lemmas required for the proof of \Cref{thm:Cond}. First, we discuss the averaging construction of \Cref{lem:alphaAnonNeutral} in detail. 

\alphaAnonNeutral*
\begin{proof}
	Let $f$ denote an arbitrary strategyproof SDS that is $\alpha$-Condorcet-consistent for some $\alpha\in [0,1]$. We construct in the sequel an anonymous and neutral SDS $f^*$ that satisfies strategyproofness and $\alpha$-Condorcet-consistency for the same $\alpha$ as $f$. As first step, we define the SDS $f^{\pi\tau}$ for arbitrary permutations $\pi:N\rightarrow N$ and $\tau:A\rightarrow A$ as follows. First, $f^{\pi\tau}$ permutes the voters in the input profile $R$ according to $\pi$ and the alternatives according to $\tau$. Next, we compute $f$ on the resulting profile $\tau(\pi(R))$ and finally, we define $f^{\pi\tau}(R,x)$ as the probability assigned to $\tau(x)$ by $f$ in $\tau(\pi(R))$. More formally, $f^{\pi\tau}$ is defined as $f^{\pi\tau}(R,x)=f(\tau(\pi(R)), \tau(x))$, where the profile $\tau(\pi(R))$ satisfies for all $i\in N$ and $x,y\in A$ that $\tau(x)\succ_{\pi(i)} \tau(y)$ in $\tau(\pi(R))$ if and only if $x\succ_i y$ in $R$. Note that $f^{\pi\tau}$ is strategyproof for all permutations $\pi$ and $\tau$ because every manipulation of $f^{\pi\tau}$ implies a manipulation of $f$. Furthermore, $f^{\pi\tau}$ is $\alpha$-Condorcet-consistent because for every preference profile $R$ with Condorcet winner $x$, $\tau(x)$ is the Condorcet winner in $\tau(\pi(R))$. Hence, if $f^{\pi\tau}$ violates $\alpha$-Condorcet-consistency in some profile $R$, then $f$ violates this axiom in the profile $\tau(\pi(R))$. 
	
	Finally, we define the SDS $f^*$ by averaging over $f^{\pi\tau}$ for all permutations $\pi$ and $\tau$. Hence, let $\Pi$ denote the set of all permutations on $N$ and let $\Tau$ denote the set of all permutations on $A$. Then, $f^*$ is defined as follows. 
	
	\begin{equation*}
		\begin{aligned}
			f^*(R, x) :=& \sum_{\pi \in \Pi} \frac{1}{|\Pi|} \sum_{\tau \in \Tau} \frac{1}{|\Tau|} f^{\pi\tau}(R, x) \\
			=& \sum_{\pi \in \Pi} \sum_{\tau \in \Tau} \frac{1}{n!m!} f(\tau(\pi(R)), \tau(x))
		\end{aligned}
	\end{equation*}
	
	Next, we show that $f^*$ satisfies all axioms required by the lemma. First, $f^*$ is strategyproof since all SDSs $f^{\pi\tau}$ are strategyproof. The $\alpha$-Condorcet-consistency of $f^*$ is shown by the following inequality, where $R$ denotes a profile in which $x$ is the Condorcet winner.
	
	\begin{equation*}
	f^*(R, x) =
	\sum_{\pi \in \Pi} \sum_{\tau \in \Tau} \frac{1}{n!m!} f(\tau(\pi(R)), \tau(x)) \ge
	\sum_{\pi \in \Pi} \sum_{\tau \in \Tau} \frac{1}{n!m!} \alpha =
	\alpha
	\end{equation*}
	
	Furthermore, observe that $f^*$ is anonymous because it averages over all possible permutations of the voters, i.e., for all permutations of the voters $\pi \in \Pi: f^*(R) = f^*(\pi(R))$. It follows from a similar argument that $f^*$ is neutral: since $f^*$ averages over all permutations of the alternatives, it holds that $f^*(R, x) =	f^*(\tau(R), \tau(x))$ for every $\tau\in\Tau$. Hence, $f^*$ is strategyproof, $\alpha$-Condorcet-consistent, anonymous, and neutral.
\end{proof}

Next, we present the proof of \Cref{lem:alphaPoint} which demonstrates that point voting SDSs cannot satisfy $\alpha$-Condorcet-consistency for $\alpha\geq \frac{2}{m}$. Note that we use additional notation for this proof. The \emph{rank} $r(x, \succeq_i)$ of an alternative $x$ in the preferences of a voter $i$ is the number of alternatives that are weakly preferred to $x$ by voter $i$, i.e., $r(x, \succeq_i)=|\{y\in A\colon y\succeq_i x\}|$. Moreover, the \emph{rank vector} $r^*(x,R)$ of an alternative $x$ in a preference profile $R$ is the vector that contains the rank of $x$ with respect to every voter in increasing order. An important observation for point voting SDSs $f$ is that $f(R,x)=f(R',x)$ if $r^*(x,R)=r^*(x,R')$. The reason for this is that a point voting SDSs assign an alternative every time probability $a_i$ when it is ranked $i$-th. Finally, the proof focuses mainly on \emph{Condorcet winner candidates}, which are alternatives that can be made into the Condorcet winner without changing their rank vectors. 

\alphaPoint*
\begin{proof}
	Let $f$ be a point voting SDS for $m\geq 3$ alternatives and $n\geq 3$ voters, and let $a=(a_1, \dots, a_m)$ be the scoring vector that defines $f$. Furthermore, assume for contradiction that $f$ is $\alpha$-Condorcet-consistent for $\alpha\geq \frac{2}{m}$. In the sequel, we show that there can be many Condorcet winner candidates in a profile $R$. Since we can turn Condorcet winner candidates into Condorcet winners without changing their rank vector and since $f(R,x)=f(R',x)$ for all profiles $R$ and $R'$ with $r^*(x,R)=r^*(x,R')$, it follows that each Condorcet winner candidate has at least probability $\alpha$ in $R$. This observation is in conflict with $\sum_{x\in A}f(R,x)=1$ if $\alpha> \frac{2}{m}$ because there can be $\lceil\frac{m}{2}\rceil$ Condorcet winner candidates. By investigating our profiles in more detail, we also deduce that $\alpha=\frac{2}{m}$ is not possible. 

	We use a case distinction with respect to the parity of $n$ and $m$ to construct profiles with $\lceil\frac{m}{2}\rceil$ Condorcet winner candidates. Moreover, we first focus on cases with fixed $n$, and provide in the end an argument for generalizing the impossibility to all $n\geq 3$. \Cref{fig:expointvoting} illustrates our construction for all four base cases with $m\in \{3,4\}$. \medskip
			
	\begin{figure}[tbp]
		\begin{tabular}{c c}
			2 & 1 \\
			\midrule
			$x_1$ & $x_2$ \\
			$x_2$ & $x_3$ \\
			$x_3$ & $x_1$ \\
			\\[1em]
			\multicolumn{2}{c}{$R^1$}
		\end{tabular}
		\quad\quad
		\begin{tabular}{c c}
			2 & 1 \\
			\midrule
			$x_1$ & $x_4$ \\
			$x_2$ & $x_2$ \\
			$x_3$ & $x_3$ \\
			$x_4$ & $x_1$ \\[1em]
			\multicolumn{2}{c}{$R^2$}
		\end{tabular}
		\quad\quad
		\begin{tabular}{c c}
			2 & 2 \\
			\midrule
			$x_1$ & $x_2$ \\
			$x_2$ & $x_1$ \\
			$x_3$ & $x_3$ \\
			\\[1em]
			\multicolumn{2}{c}{$R^3$}
		\end{tabular}
		\quad\quad
		\begin{tabular}{c c c}
			2 & 1 & 1\\
			\midrule
			$x_1$ & $x_2$ & $x_4$ \\
			$x_2$ & $x_1$ & $x_2$ \\
			$x_3$ & $x_3$ & $x_1$\\
			$x_4$ & $x_4$ & $x_3$\\[1em]
			\multicolumn{3}{c}{$R^4$}
		\end{tabular}
		\centering
		\caption{Profiles used in the base cases of the proof of \Cref{lem:alphaPoint} if $m\in {3,4}$. The profile $R^k$ shows the profile corresponding to case $k$.}
		\label{fig:expointvoting}
		\Description{The figure shows the following four preference profiles. The first profile R1 has 3 voters and 3 alternatives. Voter 1 and 2 prefer x1, x2, x3, while voter 2 prefers x2, x3, x1. The second profile R2 has 3 voters and 4 alternatives. Voter 1 and 2 prefer x1, x2, x3, x4, while voter 3 prefers x4, x2, x1, x3. The third profile R3 has 4 voters and 3 alternatives. Voter 1 and 2 prefer x1, x2, x3, while voter 3 and 4 prefer x2, x1, x3. The fourth profile R4 has 4 voters and 4 alternatives. Voter 1 and 2 prefer x1, x2, x3, x4, voter 3 prefers x2, x1, x3, x4, and voter 4 prefers x4, x2, x1, x3.}
	\end{figure}
	
	\textbf{Case 1: $n=3$ and $m$ is odd}
	
	In this case, we choose $k=\frac{m + 1}{2}$ alternatives which are denoted by $x_1, \dots, x_{k}$. 
	We construct the profile $R^1$ with $k$ Condorcet winner candidates as follows. 
	For every $i \in \{1, \dots, k\}$, voters $1$ and $2$ rank alternative $x_i$ at position $i$, and voter $3$ ranks it at position $m + 2 - 2i$. 
	The sum of ranks of $x_i$ is then equal to $2i + m + 2 - 2i = m + 2$, which means that only $m - 1$ alternatives can be ranked above $x_i$. 
	Note for this that the sum of ranks of an alternative $x$ is the number of voters $n$ plus the number of alternatives that are ranked above $x$. 
	Hence, for every $i\in \{1,\dots, k\}$, we can reorder the alternatives in $A\setminus \{x_i\}$ such that each alternative $y\in A\setminus \{x_i\}$ is preferred to $x_i$ by a single voter. 
	Consequently, $x_i$ is a Condorcet winner candidate in $R^1$, and thus $f(R^1,x_i)\geq \alpha$ for all $i\in \{1,\dots, k\}$. 
	Since there are $k=\frac{m + 1}{2}$ Condorcet winner candidates and $\sum_{i=1}^k f(R^1,x_i)\leq 1$, we derive that $\alpha \frac{m + 1}{2}\leq 1$. 
	This is equivalent to $\alpha\leq \frac{2}{m+1}<\frac{2}{m}$, which shows that $f$ fails $\alpha$-Condorcet-consistency for $\alpha\geq \frac{2}{m}$ in this case.
	\medskip
	
	\textbf{Case 2: $n=3$ and $m$ is even}
	
	If $n=3$ and $m$ is even, we construct a preference profile $R^2$ with $\frac{m}{2}$ Condorcet winner candidates similar to the last case. More precisely, we first choose an alternative $z$, and apply the construction of the last case to the alternatives $A\setminus \{z\}$. Then, we add $z$ as the last-ranked alternative of voters 1 and 2 and as first-ranked alternative of voter 3. Note that adding $z$ does not affect whether an alternative is a Condorcet winner candidate because it is last-ranked by two out of three voters. Thus, there are $\frac{m}{2}$ Condorcet winner candidates in $R^2$ and it follows analogously to the last case that $\alpha\leq \frac{2}{m}$. Finally, we show that $\alpha=\frac{2}{m}$ is also impossible. Otherwise, each of the $\frac{m}{2}$ Condorcet winner candidates has a probability of $\frac{2}{m}$, which means that the other alternatives have a probability of $0$. Thus, $f(R^2,z)=0$ even though  voter $3$ reports $z$ as his best alternative. This implies for the scoring vector $a=(a_1, \dots, a_m)$ of $f$ that $a_1=0$. However, this is not possible because the scoring vector $a$ needs to satisfy $\sum_{i=1}^m a_i=\frac{1}{n}$ and $a_i\geq a_j$ if $i\leq j$. Hence, we deduce also for this case that $\alpha<\frac{2}{m}$ holds. 
	\medskip
	
	\textbf{Case 3: $n=4$ and $m$ is odd}
	
	Just as in the first case, we choose $k=\frac{m + 1}{2}$ alternatives which are denoted by $x_1, \dots, x_{k}$. Next, we construct a profile $R^3$ with $k$ Condorcet winner candidates as follows. For every $i \in \{1, \dots, k\}$, voters 1 and 2 rank alternative $x_i$ at position $i$, and voters 3 and 4 rank it at position $\frac{m + 1}{2} + 1 - i$. The sum of ranks of $x_i$ is then equal to $2i + 2\left(\frac{m + 1}{2} + 1 - i\right) = m + 3$. Since the sum of ranks of an alternative $x$ is the number of voters plus the number of alternatives ranked above $x$, we derive that only $m - 1$ alternatives can be ranked above $x_i$. Hence, for every $i\in \{1,\dots, k\}$, we can reorder the alternatives such that each alternative $y\in A\setminus \{x_i\}$ is ranked above $x_i$ once without changing the rank vector of $x_i$. This entails that each alternative $x_i$ is a Condorcet winner candidate and thus, we derive that $\alpha\leq\frac{2}{m+1}<\frac{2}{m}$ analogously to Case 1.\medskip
	
	\textbf{Case 4: $n=4$ and $m$ is even}
	
	Finally, consider the case that $n=4$ and $m$ is even. In this situation, we construct the profile $R^4$ with $\frac{m}{2}$ Condorcet winner candidates as follows: we choose an alternative $z$, and apply the construction of Case 3 to the alternatives in $A\setminus \{z\}$. Then, voters 1 to 3 add $z$ as their least preferred alternative and voter 4 adds it as his best alternative. Just as in Case 2, every alternative that is a Condorcet winner candidate before adding $z$ is also a Condorcet winner candidate after adding this alternative because $z$ is the least preferred alternative of a majority of the voters. Hence, there are $\frac{m}{2}$ Condorcet winner candidates in $R^4$, which implies that $\alpha\leq \frac{m}{2}$. Finally, an analogous argument as in Case 2 shows that $\alpha=\frac{2}{m}$ is not possible either. In particular, if $\alpha=\frac{2}{m}$, then $f(R^4,z)=0$ because only Condorcet winner candidates can have positive probability. However, $f(R^4,z)=0$ conflicts with the definition of point voting SDSs since voter $4$ reports $z$ as his favorite choice. Therefore, it follows that $f$ fails $\alpha$-Condorcet-consistency for $\alpha\geq\frac{2}{m}$.\medskip
	
	\textbf{Case 5: Generalizing the impossibility to larger $n$}
	
	Finally, we explain how to generalize the last four cases to an arbitrary number of voters $n\geq 3$. In this case, we also construct a profile with $\lceil\frac{m}{2}\rceil$ Condorcet winner candidates. In more detail, we choose the suitable base case and add repeatedly pairs of voters with inverse preferences until there are $n$ voters. Note that voters with inverse preferences do not change the majority margins, and therefore they do not change whether an alternative is a Condorcet winner candidate. Hence, every alternative that is a Condorcet winner candidate in the base case is also a Condorcet winner candidate in the extended profile, which means that the arguments in the base cases also apply for larger numbers of voters. Therefore, no point voting SDS satisfies $\alpha$-Condorcet-consistency for $\alpha\geq \frac{2}{m}$
\end{proof}

Next, we prove \Cref{lem:alphaSup}, which bounds the probability that can be guaranteed to Condorcet winners by supporting size SDSs. 

\alphaSup*
\begin{proof}
	Let $f$ be a supporting size SDS and let $b=(b_n, \dots, b_0)$ be the scoring vector that defines $f$. Recall that the definition of a supporting size SDS requires that $b_n \ge \dots \ge b_0 \ge 0$ and $b_i + b_{n-i} = \frac{2}{m(m-1)}$ for all $i\in \{0,\dots, n\}$. In particular, this implies that $b_i \le \frac{2}{m(m-1)}$ for all $i\in \{0,\dots, n\}$. Moreover, the probability that the SDS $f$ assigns to alternative $x$ in a profile $R$ is $f(R, x) = \sum_{y \in A \setminus \{x\}} b_{n_{xy}(R)}$. Since $b_i \le \frac{2}{m(m-1)}$ for all $i\in \{0,\dots, n\}$, we derive therefore that
	$f(R, x) \le (m - 1)\frac{2}{m(m-1)} = \frac{2}{m}$ for all preference profiles $R$ and alternatives $x\in A$. 
\end{proof}

\subsection{Proofs of \Cref{lem:expost} and \Cref{thm:alphabeta}}

We focus next on the proofs of the lemmas that are required for \Cref{lem:expost}. Hence, our goal is to derive a lower bound for the $\beta$-\emph{ex post} efficiency of strategyproof $0$-randomly dictatorial SDSs. Since \Cref{thm:Gibbard1} allows us to represent strategyproof SDSs as a mixtures of duples and unilaterals, we focus next on these two classes.

First, we investigate the $\beta$-\emph{ex post} efficiency of duples. Recall therefore that a duple is a strategyproof SDS $f_{xy}$ such that $f_{xy}(R,z)=0$ for all alternatives $z \in A \setminus \{x,y\}$. Moreover, a mixture of duples $f$ is defined as $f(R,x)=\sum_{y \in A \setminus \{x\}}\lambda_{xy}f_{xy}(R,x)$, where $\lambda_{xy}=\lambda_{yx}$ denote non-negative weights that sum up to $1$. Moreover, we use in this definition that $f_{xy}=f_{yx}$. Finally, note that one duple for every pair is sufficient to represent every mixture of duples because two duples $f_{xy}$ and $f_{xy}'$ can be merged into one. 

\begin{restatable}{lemma}{dupleExPost}\label{lem:dupleExPost}
	No SDS that can be represented as a convex combination of duples satisfies $\beta$-\emph{ex post} efficiency for $\beta < \frac{1}{m}$ if $m \ge 3$.
\end{restatable}
\begin{proof}
	Let $f(R, x) = \sum_{y \in A \setminus \{x\}}\lambda_{xy}f_{xy}(R,x)$ be an SDS represented as a convex combination of duples, where $f_{xy} = f_{yx}$ is the duple SDS for the pair $x$ and $y$ and $\lambda_{xy} = \lambda_{yx}$ is the weight of $f_{xy}$. Furthermore, let $R^{x,y}$ denote a profile where all voters report $x$ as best alternative and $y$ as worst one. First, note that $f(R^{x,y},x)=f(R^{x,z},x)$ and $f(R^{y,x},x)=f(R^{z,x},x)$ for all distinct $x,y,z\in A$. Thus, we also write $R^{x,\cdot}$ and $R^{\cdot, x}$ to indicate that alternative $x$ is unanimously top-ranked or bottom-ranked. 
	
	As first step, we want to bound the average probability $f(R^{x,y},x)+f(R^{x,y},y)$ over all $x,y\in A$. In more detail, the subsequent equation shows that $\sum_{x \in A}\sum_{y \in A \setminus \{x\}} \Big(f(R^{xy}, x) + f(R^{xy}, y)\Big)= 2(m - 1)$. 
	
	\begin{equation*}
		\begin{aligned}
			&\sum_{x \in A}\sum_{y \in A \setminus \{x\}} f(R^{x,y}, x) + f(R^{x,y}, y) \\
			&=(m-1)\sum_{x \in A} f(R^{x,\cdot}, x) + (m-1)\sum_{y \in A} f(R^{\cdot, y}, y) \\
			&= (m-1)\sum_{x\in A}\sum_{y\in A\setminus \{x\}} \lambda_{xy} f_{xy}(R^{x,y},x) + \lambda_{xy}f_{xy}(R^{x,y},y)\\
			&= (m-1)\sum_{x\in A}\sum_{y\in A\setminus \{x\}} \lambda_{xy}\\
			&= 2(m - 1)
		\end{aligned}
	\end{equation*}
	
	The first equality follows from $f(R^{x,y},x)=f(R^{x,\cdot},x)$, $f(R^{x, y},y)=f(R^{\cdot,y},y)$ for all alternatives $x,y\in A$, and the observation that every alternative $x$ is both unanimously top-ranked and unanimously bottom-ranked in exactly $(m-1)$ of the considered preferences profiles. For the second equality, we replace $f(R^{x,\cdot}, x)$ with $\sum_{y \in A \setminus \{x\}}\lambda_{xy}f_{xy}(R^{x, y},x)$ and $f(R^{\cdot, y}, y)$ with $\sum_{x \in A \setminus \{y\}}\lambda_{xy}f_{xy}(R^{x, y},y)$ according to the definition of $f$. Furthermore, we swap the order of the sum for the second term.
	We derive the third equality from the fact that $f_{xy}(R, x) + f_{xy}(R, y) = 1$ for all profiles $R$. Finally, the last equality uses that $ \sum_{x \in A} \sum_{y \in A \setminus \{x\}} \lambda_{xy} = 2$, which follows from $\sum_{x\in A}f(R, x) = \sum_{x\in A}\sum_{y \in A \setminus \{x\}}\lambda_{xy}f_{xy}(R, x) = 1$ and $f_{xy}(R, x) + f_{xy}(R, y) = 1$ for all profiles $R$. 
	
	As a consequence of this observation, it follows that there is a pair of alternatives $x,y\in A$ such that $f(R^{x,y},x)+f(R^{x,y},y)\leq \frac{2}{m}$. Otherwise, it holds that $\sum_{x \in A}\sum_{y \in A \setminus \{x\}} f(R^{x,y}, x) + f(R^{x,y}, y)>\sum_{x \in A}\sum_{y \in A \setminus \{x\}} \frac{2}{m}=2(m-1)$, contradicting our previous equation. Hence, $\sum_{z\in A\setminus \{x,y\}} f(R^{x,y},z)\geq \frac{m-2}{m}$. Since all alternatives $z\in A\setminus \{x,y\}$ are Pareto-dominated by $z$, this entails that one of these alternative receives a probability of at least $\frac{m - 2}{m(m - 2)} = \frac{1}{m}$. We conclude therefore that the SDS $f$ fails $\beta$-ex post efficient for $\beta < \frac{1}{m}$.
\end{proof}

Next, we aim to show that no $0$-randomly dictatorial SDS that can be represented as a mixture of unilaterals satisfies $\beta$-\emph{ex post} efficiency for $\beta<\frac{1}{m}$. Ideally, we would like to use $f$ to construct a $0$-randomly dictatorial SDS $f^*$ that satisfies $\beta$-\emph{ex post} efficiency for the same $\beta$ as $f$, and that is additionally neutral and anonymous. Unfortunately, we cannot use \Cref{lem:alphaAnonNeutral} here as this lemma does not preserve that $f^*$ is $0$-randomly dictatorial. For demonstrating this point, let $A=\{x_1, \dots, x_m\}$ denote the alternatives and assume that $m=n\geq 3$. Furthermore, consider the unilateral $f^i$ which assigns probability $1$ to voter $i$'s favorite alternative in $A\setminus \{x_i\}$. Finally, consider the SDS $f^+$ which chooses a voter $i\in N$ uniformly at random and returns the outcome of $f^i$. \Cref{lem:zerodict} shows that this SDS is $0$-randomly dictatorial because for all $i\in N$, the probability of $x_i$ does not increase if voter $i$ reinforces it to his best alternative. Moreover, since $f^+$ is a mixture of unilaterals, it is strategyproof, and its definition implies that it not anonymous. However, applying the construction of \Cref{lem:alphaAnonNeutral} to $f^+$ results in the point voting SDS defined by the scoring vector $(\frac{m-1}{nm}, \frac{1}{nm}, 0, \dots, 0)$. It follows immediately from \Cref{lem:zerodict} that this SDS is not $0$-randomly dictatorial as pushing an alternative from second place to first place increases its probability always by $\frac{m-2}{nm}>0$. 

Therefore, we propose another construction in the next lemma that, given an arbitrary strategyproof and $0$-randomly mixture of unilaterals, constructs a strategyproof $0$-randomly dictatorial SDS that is $\beta$-\emph{ex post} efficient for the same $\beta$ as the original SDS and that has a lot of symmetries. Unfortunately, this construction does not result in a anonymous SDS. Nevertheless, the resulting SDS is significantly easier to work with and its properties are crucial for the proof of \Cref{lem:expostUnilateral}. Note that we require some additional terminology for the next lemma. In the sequel, we say that voter $i$ or his unilateral SDS $f_i$ is $0$-randomly dictatorial for alternatives $x,y$ if $f(R)=f(R^{i:yx})$ for \emph{all} preference profiles $R$ in which $x$ is voter $i$'s best alternative and $y$ is his second best alternative.

\begin{lemma}\label{lem:expostconst}
	Let $f$ be a strategyproof $0$-randomly dictatorial SDS that satisfies $\beta$-\emph{ex post} efficiency for some $\beta\in [0,1]$ and that can be represented as a mixture of unilaterals. Then, there is a strategyproof $0$-randomly dictatorial SDS $f^*$ for $\binom{m}{2}$ voters that can be represented as a mixture of unilaterals and that is $\beta$-\emph{ex post} efficient for the same $\beta$ as $f$. Moreover, $f^*$ satisfies the following conditions:
	\begin{enumerate}[label=(\roman*)]
		\item For every voter $i\in N$, there is a set $\{x_i, y_i\}$ such that voter $i$ is $0$-randomly dictatorial for $x_i$, $y_i$ and $\{x_i,y_i\}\neq \{x_j,y_j\}$ if $i\neq j$.
		\item There is a constant $\delta$ such that $f^*(R^{i:cb},c)-f^*(R,c)=\delta$ for all voters $i\in N$, alternatives $\{a,b\}=\{x_i, y_i\}$, $c\in A\setminus\{x_i,y_i\}$, and preference profiles $R$ such that voter $i$ reports $a$ as his best alternative, $b$ as his second best one, and $c$ as his third best one. 
		\item If every voter $i\in N$ reports $x_i$ and $y_i$ as their two best alternatives, then there exists a scoring vector $a=(a_1, \dots, a_m)$ such that $a_1 = a_2 \ge 0$, $a_3 \ge \dots \ge a_m \ge 0$, and $f^*(R, x) = \sum_{i \in N} a_{|\{y \in A: y \succeq_i x\}|}$.
	\end{enumerate}
\end{lemma}
\begin{proof}
	Let $\beta \in [0,1]$ and let $f$ denote a strategyproof $0$-randomly dictatorial SDS that is $\beta$-\emph{ex post} efficient and that can be represented as a mixture of unilaterals. In the sequel, we use $f$ to construct the SDS $f^*$ that satisfies all requirements of the lemma. Note that this proof is quite involved and therefore, we use some auxiliary observations that are proven in the end.
	
	We start by representing $f$ as $f(R) = \sum_{i \in N} \lambda_i f_i(\succeq_i)$, where $f_i$ denotes the unilateral SDS of voter $i$ and $\lambda_i\geq 0$ is its weight. Note that we interpret unilaterals in this proof as SDSs that take a single preference relation as input. This is possible as unilaterals only rely on the preferences of a single voter. Observation $1$ states that for every voter $i\in N$ there are alternatives $x_i$, $y_i$ such that $f_i$ is $0$-randomly dictatorial for $x$ and $y$. Even though a voter can be $0$-randomly dictatorial for multiple pairs of alternatives, we associate from now on every voter $i$ with exactly one such pair $x_i,y_i$. This pair can be chosen arbitrarily as it will not affect the rest of the proof.
	
	Next, we define the unilaterals $f_i^\tau$ as $f^\tau_i(R,x)=f_i(\tau(R), \tau(x))$ for all voters $i\in N$ and permutations $\tau:A\rightarrow A$. Observation $2$ states that every SDS $f_i^\tau$ is strategyproof and $0$-randomly dictatorial for $\tau^{-1}(x_i)$, $\tau^{-1}(y_i)$, where $\tau^{-1}$ is the inverse permutation of $\tau$ and $x_i$ and $y_i$ are the alternatives associated with $f_i$. Just as the SDSs $f_i$, each $f_i^\tau$ can be $0$-randomly dictatorial for multiple pairs of alternatives, but we associate $f^\tau_i$ from now on only with the pair $\tau^{-1}(x_i)$, $\tau^{-1}(y_i)$. Then, we partition the SDSs $f^\tau_i$ with respect to the alternatives $\tau^{-1}(x_i)$, $\tau^{-1}(y_i)$. In more detail, let $F_{xy}=\{f^\tau_i\colon i\in N, \tau\in \Tau, \{\tau^{-1}(x_i), \tau^{-1}(y_i)\}=\{x,y\}\}$ denote the multi-set of SDSs $f^{\tau}_i$ that are associated with $x$ and $y$. Note that all unilaterals in $F_{xy}$ are $0$-randomly dictatorial for $x,y$. Furthermore, these multi-sets indeed partition the SDSs $f^\tau_i$ as each $f^\tau_i$ is only associated with a single pair of alternatives. Even more, there are for every $f_i$ exactly $2(m-2)!$ permutations $\tau$ such that $\{\tau^{-1}(x_i), \tau^{-1}(y_i)\}=\{x,y\}$. Hence, we derive that each set $F_{xy}$ contains $2n(m-2)!$ SDSs. 
	
	In the next step, we merge all unilaterals in a multi-set $F_{xy}$ into a single unilateral. Thus, we define the unilateral $f_{xy}(\succeq_j)$ as $f_{xy}(\succeq_j)=\sum_{f^\tau_i\in F_{xy}} \frac{\lambda_i}{2(m-2)!} f^\tau_i(\succeq_j)$, i.e., $f_{xy}$ chooses each SDS $f^\tau_i\in F_{xy}$ with a probability proportional to $\lambda_i$. Observe that $f_{xy}$ is strategyproof because it is a mixture of strategyproof SDSs and it is $0$-randomly dictatorial for $x,y$ because all unilaterals in $F_{xy}$ are $0$-randomly dictatorial for these alternatives. Based on the SDS $f_{xy}$, we can finally define the SDS $f^*$ for $n^*=\binom{m}{2}$ voters. To this end, let $N^*$ denote the electorate of $f^*$. We associate each voter $j\in N^*$ with a different pair of alternatives $x,y\in A$ and set $f^*_j=f_{xy}$. Then, the SDS $f^*$ chooses one of the voters $j\in N^*$ uniformly at random and returns $f^*_j(\succeq_j)=f_{xy}(\succeq_j)$, i.e., $f^*(R)=\frac{1}{n^*}\sum_{j=1}^{n^*}f_j^*(\succeq_j)$. Clearly, $f^*$ is strategyproof because it is a mixture of strategyproof SDSs. Moreover, it is $0$-randomly dictatorial because every voter $j\in N^*$ is $0$-randomly dictatorial for the pair of alternatives $x,y$ with which he is associated. Furthermore, Observation 3 shows that $f^*$ is $\beta$-\emph{ex post} efficient for the same $\beta$ as $f$.
	
	It remains to show that the SDS $f^*$ satisfies the properties \emph{(i)}, \emph{(ii)}, and \emph{(iii)}. First, note that it satisfies \emph{(i)} by construction as every voter is $0$-randomly dictatorial for a different pair of alternatives. For \emph{(ii)} and \emph{(iii)}, we show first the auxiliary claim that $f_{xy}(R,x)=f_{\tau(x)\tau(y)}(\tau(R),\tau(x))$ for all permutations $\tau:A\rightarrow A$, preference profiles $R$, and alternatives $x\in A$. Note that if this claim holds then the SDS $f^*$ satisfies neutrality, since then for all permutations $\tau \in \Tau$ and alternatives $x \in A: f^*(\tau(R), \tau(x)) = \frac{1}{n^*}\sum_{j=1}^{n^*}f_{\tau(x_j)\tau(y_j)}(\tau(\succeq_j), \tau(x)) = \frac{1}{n^*}\sum_{j=1}^{n^*}f_{x_jy_j}(\succeq_j, x) = f^*(R, x)$. Hence, consider an arbitrary SDS $f^{\tau'}_i\in F_{xy}$ and note that $f_i^{\tau'}(R,x)=f_i(\tau'(R), \tau'(x)) = f_i^{\tau'\circ \tau^{-1}}(\tau(R), \tau(x))$. Next, observe that $f^{\tau'\circ\tau^{-1}}\in F_{\tau(x),\tau(y)}$. This is true because $f^{\tau'}_i\in F_{xy}$ implies that $\{\tau'(x), \tau'(y)\}=\{x_i, y_i\}$. Therefore, $\{\tau'(\tau^{-1}(\tau(x))), \tau'(\tau^{-1}(\tau(y)))\}=\{x_i, y_i\}$ which shows that $f^{\tau'\circ\tau^{-1}}\in F_{\tau(x)\tau(y)}$. Finally, we derive the following equality for all profiles $R$ and alternatives $x\in A$.
	\begin{align*}
		f_{xy}(R,x)&=\sum_{f_i^\tau\in F_{xy}} \frac{\lambda_i}{2(m-1)!} f_i^\tau(R,x)\\
		&=\sum_{f_i^\tau\in F_{xy}} \frac{\lambda_i}{2(m-1)!} f_i^{\tau'\circ \tau^{-1}}(\tau(R),\tau(x))\\
		&=\sum_{f_i^{\hat \tau}\in F_{\tau(x)\tau(y)}} \frac{\lambda_i}{2(m-1)!} f_i^{\hat \tau}(\tau(R),\tau(x))\\
		&=f_{\tau(x)\tau(y)}(\tau(R),\tau(x)). 
	\end{align*}
	
	In the second step of this equation, we define $\hat \tau = \tau'\circ \tau^{-1}$. The third step uses the fact that $\tau'\circ \tau^{-1}\neq \tau''\circ \tau^{-1}$ if $\tau'\neq\tau''$, which implies that every SDS $f^\tau_i\in F_{xy}$ is mapped to a unique element in $F_{\tau(x)\tau(y)}$. This proves the auxiliary claim.
	
	Subsequently, we show \emph{(ii)} and consider therefore an arbitrary voter $i\in N^*$. Moreover, let $x_i$, $y_i$ denote the alternatives associated with $f^*_i$, i.e., $f^*_i=f_{x_iy_i}$. Finally, let $z_i\in A\setminus \{x\}$ denote a third alternative and consider a profile $R$ in which voter $i$ prefers $x_i$ the most, $y_i$ the second most, and some arbitrary alternative $z_i\in A\setminus \{x_i, y_i\}$ the third most. We define $\delta=f^*(R^{i:z_iy_i},z)-f^*(R,z_i)$. First, note that $R^{i:z_iy_i}$ and $R$ only differ in the preferences of voter $i$ and thus, $f^*(R^{i:z_iy_i},z_i)-f^*(R,z_i)=f^*_i(\succeq^{i:z_iy_i}_i,z_i)-f^*_i(\succeq_i,z_i)$. Next, consider a second voter $j\in N^*$ ($j=i$ is possible), let $x_j$ and $y_j$ denote the alternatives which are associated with $f^*_j$, and let $z_j\in A\setminus \{x_j,y_j\}$ denote another alternative. Finally, consider a profile $R'$ such that voter $j$ ranks $x_j$ first, $y_j$ second, and $z_j$ third in $R'$, and define $R^+=(R')^{j:z_jy_j}$. We show in the sequel that $f^*(R^+,z_j)-f^*(R',z_j)=\delta$, which proves claim \emph{(ii)}. Thus, note first that $f^*(R^+,z_j)-f^*(R',z_j)=f^*_j(\succeq^+_j,z_j)-f^*_j(\succeq_j',z_j)$ because $f^*$ is a mixture of unilaterals and only voter $j$ changes his preference relation. Next, let $\tau$ denote a permutation such that $\tau(\succeq_j')=\tau(\succeq_i)$, which means in particular that $\tau(x_j)=x_i$, $\tau(y_j)=y_i$, and $\tau(z_j)=z_i$. Now, our auxiliary claim proves \emph{(ii)} since
	\begin{align*}\allowdisplaybreaks
		&f^*_j(\succeq^+_j,z_j)-f^*_j(\succeq'_j,z_j)\\
		=&f_{x_jy_j}((\succeq')^{j:z_jy_j}_j,z_j)-f_{x_jy_j}(\succeq'_j,z_j)\\
		=&f_{\tau(x_j)\tau(y_j)}(\tau((\succeq')^{j:z_jy_j}_j), \tau(z_j)) -  f_{\tau(x_j)\tau(y_j)}(\tau(\succeq_j'), \tau(z_j))\\
		=&f_{x_iy_i}(\succeq^{i:z_iy_i}_i, z_i) -  f_{x_iy_i}(\succeq_i, z_i)\\
		=&\delta.
	\end{align*}
	
	Finally, we discuss why property \emph{(iii)} is true. Thus, consider two voters $i,j\in N^*$ and let $x_i, y_i$ and $x_j$, $y_j$ denote the alternatives associated with $f^*_i$ and $f^*_j$, respectively. We explicitly allow in the subsequent analysis that $i=j$. Furthermore, consider two preference relations $\succeq_i$ and $\succeq_j$ such that $x_i$ and $y_i$ are top-ranked in $\succeq_i$ and $x_j$ and $y_j$ are top-ranked in $\succeq_j$. Finally, let $\tau$ denote a permutation such that $\succeq_i=\tau(\succeq_j)$ and let $z^k_i$ and $z^k_j$ denote the $k$-th ranked alternative of voter $i$ and $j$, respectively. Our auxiliary claim shows immediately that $f^*_i(\succeq_i, z^k_i)=f^*_j(\succeq_j, z^k_j)$. This means that, for every $k\in \{1,\dots,m\}$, the $k$-th ranked alternative receives the same probability from every voter if they report the alternatives $x_i$, $y_i$ as their favorite choice. Hence, there is a scoring vector $a=(a_1,\dots, a_m)$ such that $f^*(R,x)=\sum_{i \in N} a_{|\{y \in A: y \succeq_i x\}|}$ for such profiles. Moreover, it follows from strategyproofness that $a_3\geq a_4\geq \dots a_m$ and from the definition of an SDS that $a_i\geq 0$ for all $i\in \{1,\dots, m\}$. Finally, $a_1=a_2$ since for all $i\in N^*$, the unilateral $f^*_i$ is $0$-randomly dictatorial for $x_i$ and $y_i$. Hence, there is a scoring vector that meets all requirements of \emph{(iii)}.
	\medskip

	\textbf{Observation 1: For every voter $i$, there exists a pair of alternatives $x_i$, $y_i$ such that $f(R)=f(R^{i:y_ix_i})$ for all preference profiles $R$ in which voter $i$ reports $x$ as best alternative and $y$ as second best one.}

	Since $f_i$ is a strategyproof $0$-randomly dictatorial SDS, it follows from \Cref{lem:zerodict} that for every voter $i\in N$, there exists a pair of alternatives $x_i, y_i$ and a preference profile $R$ such that $f_i(R,y)=f_i(R^{i:y_ix_i},y)$, voter $i$ top-ranks $x_i$ $R$, and second-ranks $y_i$. First, note that localizedness immediately generalizes this observation to $f_i(R)=f_i(R^{i:y_ix_i})$. We show in the sequel that $f(R)=f(R^*)$ for all preference profiles $R$, $R^*$ such that voter $i$ reports $x_i$ and $y_i$ as his best and second best alternative in $R$ and $R^*=R^{i:yx}$. Since $f$ is a mixture of strategyproof unilaterals, it follows that $f(R)=f(R^*)$ if $f_i(\succeq_i)=f(\succeq^*_i)$ because $\succeq_j=\succeq_j^*$ for all $j\in N\setminus \{i\}$. Moreover, it follows from strategyproofness, which entails localizedness, that $f_i(\succeq'_i,z)=f_i(\succeq_i,z)=f_i(\succeq^*_i,z)=f_i(\succeq^+_i,z)$ for $z\in \{x_i,y_i\}$ and all preferences profiles $R'$ and $R^+$ such $R'=(R^+)^{i:y_ix_i}$ and $\succeq_i'$ only differs from $\succeq_i$ in the order of the alternatives $A\setminus \{x_i,y_i\}$. Because $R'$ and $R^+$ differ by definition only in voter $i$'s preference over $x_i$ and $y_i$, another application of localizedness implies that $f_i(R')=f_i(R^+)$. Hence, it holds indeed that $f(R)=f(R^{i:y_ix_i})$ for all preference profiles in which voter $i$ reports $x_i$ and $y_i$ as his two best alternatives.\medskip
	
	\textbf{Observation 2: The SDS $f^\tau_i(R,x)=f_i(\tau(R), \tau(x))$ is strategyproof and $0$-randomly dictatorial for $\tau^{-1}(x_i)$, $\tau^{-1}(y_i)$.}
	
	First, note that $f^\tau_i$ is strategyproof as every manipulation of this SDS could be mapped to a manipulation of $f_i$. In more detail, if voter $i$ can manipulate $f^\tau_i$ by switching from $R$ to $R'$, he can also manipulate $f_i$ by switching from $\tau(R)$ to $\tau(R')$. This is true because a manipulation requires an alternative $x$ such that $\sum_{y\succ_i x} f^\tau_i(R',y)>\sum_{y\succ_i x} f^\tau_i(R,y)$, which entails by definition of $f^\tau_i$ that $\sum_{y\succ_i x} f_i(\tau(R'),\tau(y))>\sum_{y\succ_i x} f_i(\tau(R),\tau(y))$. Finally, since $y\succ_i x$ in $R$ if and only if $\tau(y)\succ_i \tau(x)$ in $\tau(R)$, we derive that voter $i$ could manipulate $f_i$ by switching from $\tau(R)$ to $\tau(R')$. 
	
	Furthermore, $f^\tau_i$ is a $0$-randomly dictatorial SDS because $f_i$ is one: Observation 1 shows that for every voter $i$, there exists a pair of alternatives $x_i$, $y_i$ such that $f(R)=f(R^{i:y_ix_i})$ for all preference profiles $R$ in which voter $i$ prefers $x_i$ the most and $y_i$ the second most. It follows from this observation that $f^\tau_i(\tau^{-1}(R),\tau^{-1}(x))=f_i(R,x)=f_i(R^{i:y_ix_i},x)=f^\tau_i(\tau^{-1}(R^{i:y_ix_i}), \tau^{-1}(x))$ for all $x\in A$, where $\tau^{-1}$ is the inverse permutation of $\tau$, i.e., $\tau^{-1}(\tau(x))=x$ for all $x\in A$. Therefore, $f^\tau_i(\tau^{-1}(R), \tau^{-1}(x_i)) = f^\tau_i(\tau^{-1}(R^{i:y_ix_i}), \tau^{-1}(x_i))$ and $f^\tau_i(\tau^{-1}(R), \tau^{-1}(y_i)) = f^\tau_i(\tau^{-1}(R^{i:y_ix_i}), \tau^{-1}(y_i))$. Moreover, the preference profiles $\tau^{-1}(R)$ and $\tau^{-1}(R^{i:y_ix_i})$ only differ in the order of the two best alternatives $\tau^{-1}(x)$ and $\tau^{-1}(y)$ of voter $i$ and the proof of Observation 1 entails thus that $f^\tau_i$ is $0$-randomly dictatorial for these two alternatives.\medskip
	
	\textbf{Observation 3: The SDS $f^*=\frac{1}{n^*} \sum_{i=1}^{n^*} f^*_i$ is $\beta$-\emph{ex post} efficient for the same $\beta$ as $f$.}
	
	For proving this observation, we construct first another SDS $f^+$ and show that this SDS is $\beta$-\emph{ex post} efficient for the same $\beta$ as $f$. As second step, we show that $f^*$ can also be derived from $f^+$ by merging voters, and thus $f^*$ inherits the $\beta$-\emph{ex post} efficiency of $f^+$. Before defining $f^+$, we introduce the SDS $f^\tau$: just as the SDSs $f^\tau_i$, it is defined as $f^\tau(R,x)=f(\tau(R), \tau(x))$. In particular, $f^\tau$ is $\beta$-\emph{ex post} efficient for the same $\beta$ as $f$. This follows by considering an arbitrary profile $R$ in which an alternative $x$ is Pareto-dominated. It is easy to see that $\tau(x)$ is then Pareto-dominated in $\tau(R)$, and we derive therefore that $f^\tau(R,x)=f(\tau(R),\tau(x))\leq \beta$ because $f$ is $\beta$-\emph{ex post} efficient. Next, we define the SDS $f^+$ for $nm!$ voters as follows: we partition the voters $\{1,\dots, nm!\}$ into $m!$ sets $N_1,\dots, N_{m!}$ with $|N_i|=n$ and associate with every set a different permutation $\tau_i:A\rightarrow A$. Then, $f^+(R)=\frac{1}{m!}\sum_{i=1}^{m!} f^{\tau_i}(\succeq_{N_i})$, where $\succeq_{N_i}$ denotes the restriction of $R$ to the voters in $N_i$. Observe that $f^+$ is $\beta$-\emph{ex post} efficient for the same $\beta$ as $f$ because an alternative $x$ that is Pareto-dominated in $R$ is also Pareto-dominated in all $\succeq_{N_i}$ and all $f^{\tau_i}$ are $\beta$-\emph{ex post} efficient. Hence, it follows that $f^+(R,x)=\frac{1}{m!}\sum_{i=1}^{m!} f^{\tau_i}(\succeq_{N_i},x)\leq \frac{1}{m!}\sum_{i+1}^{m!}\beta=\beta$.
	
		Next, we show that $f^+$ and $f^*$ satisfy $\beta$-\emph{ex post} efficiency for the same $\beta$. Therefore, we change the representation of $f^+$. The central observation here is that $f^\tau=\sum_{i\in N} \lambda_i f^\tau_i$. Hence, we can also associate every voter $j\in \{1,\dots, nm!\}$ with an index $i\in N$ and a permutation $\tau$ such that each index-permutation pair is assigned exactly once. Thus, define $f_j^+=f^\tau_i$ and $\lambda_j^+=\frac{\lambda_i}{m!}$ (i.e., the weight of $f^\tau_i$ is the proportional to the weight of $f_i$ in the original SDS $f$). Then, we can write $f^+$ as $f^+(R)=\sum_{j=1}^{nm!} \lambda_j^+ f_j^+(\succeq_j)$. Next, note that every $f^\tau_i$ appears once in $f^+(R,x)$ and once in the union of all $F_{xy}$. Therefore, we derive that $f^+(R)=\frac{1}{n^*}\sum_{\{x,y\}\subset {A\choose 2}} \sum_{f^\tau_i\in F_{xy}} \frac{\lambda_i}{2(m-2)!} f^\tau_i(\succeq_i)$, where $n^*={m\choose2}$. Next, we restrict our attention to profiles $R$ such that for all $\{x,y\}\subset {A\choose 2}$, all voters $j$ with $f_j\in F_{xy}$ submit the same preference relation. In this case, we may replace the preferences of all voters $j$ with $f_j\in F_{xy}$ with a single preference relation. Then, there are exactly $m\choose 2$ voters left, each of which is associated with a different pair of alternatives. In particular, we can use the definition of $f_{xy}(\succeq_i)=\sum_{f^\tau_i\in F_{xy}}\frac{\lambda_i}{2(m-2)!} f^\tau_i(\succeq_i)$ now as we apply all unilateral SDSs in $F_{xy}$ to the same preference relation $\succeq_i$. Hence, $f^+$ returns the same outcomes as $f^*$ if for each $\{x,y\}\subset {A\choose 2}$, all voters $j$ with $f_j\in F_{xy}$ report the same preferences. Since $f^+$ is $\beta$-\emph{ex post} efficient, it follows therefore also that $f^*$ is $\beta$-\emph{ex post} efficient.
\end{proof}

Finally, we use \Cref{lem:expostconst} to prove that no $0$-randomly dictatorial SDS that can be represented as a mixture of unilaterals is $\beta$-\emph{ex post} efficient for $\beta<\frac{1}{m}$. 

\begin{restatable}{lemma}{expostUnilateral}\label{lem:expostUnilateral}
	No $0$-randomly dictatorial SDS that can be represented as a convex combination of unilaterals satisfies $\beta$-\emph{ex post} efficiency for $\beta < \frac{1}{m}$ if $m \ge 3$.
\end{restatable}
\begin{proof}
	Let the SDS $f$ denote a mixture of unilaterals. First, we apply \Cref{lem:expostconst} to construct the SDS $f^*$ as specified by this lemma. In the sequel, we show that $f^*$ is $\beta$-\emph{ex post} efficient for $\beta \ge \frac{1}{m}$ and therefore $f$ is also $\beta$-\emph{ex post} efficient for $\beta \geq \frac{1}{m}$. In our proof, we construct a profile $R^*$ in which every alternative must receive a probability of at most $\beta$ which leads to a contradiction if $\beta < \frac{1}{m}$. Let $N$ with $|N| = \binom{m}{2}$ be the set of voters of $f^*$. Furthermore, \Cref{lem:expostconst} \emph{(i)} states that every voter $j \in N$ is associated with a different pair of alternatives $\{x_j, y_j\}$ for which he is $0$-randomly dictatorial.
	
	First, we explain the construction of an auxiliary profile $R$. For this profile, we choose an arbitrary pair of alternatives $a$, $b$ and assume without loss of generality that voter $1$ is $0$-randomly dictatorial for $a,b$, i.e, $\{a,b\}=\{x_1, y_1\}$. Voter $1$ submits the preference relation $\succeq_1=b \succ_1 a \succ_1 \dots$ in $R$. Furthermore, there are $m - 2$ other voters $j\in N$ with $a \in \{x_j, y_j\}$ and $b \notin \{x_j, y_j\}$. We assume without loss of generality that these are the voters in $\{2, \dots, m-1\}$ and that $a = x_j$. The preferences of the voters $j\in \{2,\dots, m-2\}$ in $R$ is $\succeq_j= y_j \succ_j a \succ_j b \succ_j \dots$. Also, there are $m - 2$ voters $j$ with $a \notin \{x_j, y_j\}$ and $b \in \{x_j, y_j\}$. We assume that these voters are the ones in $\{m,\dots,2m-3 \}$ and that $b=y_j$. The preferences of these voters is $\succeq_j= b \succ_j x_j \succ_j a \succ_j \dots$. Finally, the remaining voters $j\in \{2m-2, \dots, \binom{m}{2}\}$ have $a, b \notin \{x_j, y_j\}$. These voters report $\succeq_j=x_j \succ_j y_j \succ_j b \succ_j a$ in $R$. Note that if $m = 3$, there are no voters of the fourth type. Furthermore, every voter $j\in N$ ranks the alternatives $x_j, y_j$ for which he is $0$-randomly dictatorial at the top. The full profile for $m = 4$ is shown in \Cref{fig:expostUnilateralProfile1}.
	
\begin{figure}[tbp]
	\begin{tabular}{c c c c c c}
		1 & 1 & 1 & 1 & 1 & 1 \\
		\midrule
		$b$ & $c$ & $d$ & $b$ & $b$ & $c$ \\
		$a$ & $a$ & $a$ & $c$ & $d$ & $d$ \\
		$c$ & $b$ & $b$ & $a$ & $a$ & $b$ \\
		$d$ & $d$ & $c$ & $d$ & $c$ & $a$ 
	\end{tabular}
	\centering
	\caption{The preference profile $R$ that in the proof of \Cref{lem:expostUnilateral} for $m = 4$. There are four groups of voters. The first group contains only the first voter who is $0$-randomly dictatorial for $a$ and $b$. The next two groups have both $m-2$ voters and are $0$-randomly dictatorial for one of $a$ and $b$. The last group contains the remaining $\binom{m - 2}{2}$ voters that are not $0$-randomly dictatorial $a$ or $b$. All voters have the pair for which they are $0$-randomly dictatorial ranked at the top.}
	\label{fig:expostUnilateralProfile1}
	\Description{Preference Profile with 6 voters and 4 alternatives. Every voter has the pair of alternatives for which his unilateral is 0-dictatorial ranked at the top.}
\end{figure}
	
	We show next that $f^*(R,a) \le \beta$ by constructing a new preference profile $R'$ such that $f^*(R, a) = f^*(R', a) \le \beta$. For the construction of $R'$, let all voters in the second group $j \in \{2,\dots, m-1\}$ swap $a$ and $b$, and all voters in the third group $j \in \{m,\dots, 2m-3\}$ swap $a$ and $x_j$. The resulting preference profile is shown in \Cref{fig:expostUnilateralProfile2} for the case that $m=4$. It is easy to see that $b$ Pareto-dominates $a$ in $R'$ and, as $f^*$ is $\beta$-\emph{ex post} efficient, $f^*(R',a) \le \beta$. Alternative $a$ was moved from third to second and from second to third place by $m-2$ voters. It follows therefore from \Cref{lem:expostconst} \emph{(ii)} and localizedness that the probability that alternative $a$ gains when $m-2$ voters swap it from third to second place is the same as the probability that $a$ looses when $m-2$ voters swap it from second to third place. Thus, we derive that $f^*(R, a) = f^*(R',a) \le \beta$.
	
	\begin{figure}[tbp]
		\begin{tabular}{c c c c c c}
			1 & 1 & 1 & 1 & 1 & 1 \\
			\midrule
			$b$ & $c$ & $d$ & $b$ & $b$ & $c$ \\
			$a$ & $b$ & $b$ & $a$ & $a$ & $d$ \\
			$c$ & $a$ & $a$ & $c$ & $d$ & $b$ \\
			$d$ & $d$ & $c$ & $d$ & $c$ & $a$
		\end{tabular}
		\centering
		\caption{The preference profile $R'$ for $m=4$ alternatives that results from $R$ by swapping the second and third alternative of voters $j \in \{2,\dots, 2m-3\}$. Alternative $a$ is Pareto-dominated by alternative $b$.}
		\label{fig:expostUnilateralProfile2}
		\Description{Preference Profile with 6 voters and 4 alternatives. This profile is similar to the last profile except voters 2, 3, 4 and 5 swapped their second and third most preferred alternative.}
	\end{figure}
	
	Finally, note that in $R$, all voters $j\in N$ report the pair $x_j, y_j$ for which they are $0$-randomly dictatorial as their two best alternatives. Hence, \Cref{lem:expostconst} \emph{(iii)} entails the existence of a scoring vector $(a_1, \dots, a_m)$ such that $a_1 = a_2 \ge 0$, $a_3 \ge \dots \ge a_m \ge 0$, and $f^*(R, x) = \sum_{j \in N}a_{|\{y \in A: y \succeq_j x \}|}$ for all $x\in A$. In particular, observe that the probability of an alternative only depends on its rank vector $r^*(x,R)$. Recall that the rank vector $r^*(x,R)$ of an alternative $x$ in a preference profile $R$ is the vector that contains the rank of $x$ with respect to every voter in increasing order. The rank vector of alternative $a$ in $R$ is
	\begin{equation*}
	r^*(a, R) = (\overbrace{2, \dots, 2}^{m - 1}, \overbrace{3, \dots, 3}^{m-2}, \overbrace{4, \dots, 4}^{\binom{m - 2}{2}}).
	\end{equation*}
	
	Furthermore, observe that $f^*(\bar{R},x)\leq f^*(R,a)$ in every profile $\bar{R}$ in which (i) each voter $j\in N$ reports the alternatives $x_j$, $y_j$ as his two best alternatives and (ii) $r^*(x,\bar{R})_k\geq r^*(a,R)_k$ for all $k\in \{m, \dots, \binom{m}{2}\}$. Condition (i) implies that $f^*$ can be computed based on the scoring vector $(a_1, \dots, a_m)$. Furthermore, it implies that every alternative $x\in A$ is among the two best alternatives of exactly $m-1$ voters, and since $a_1=a_2$, it follows that we can ignore these entries when comparing the probability of $a$ in $R$ with the probability of $x$ in $\bar{R}$. Finally, the claim follows as $a_3\geq \dots\geq a_m$ and $r^*(x,\bar{R})_k\geq r^*(a,R)_k$ for all $k\in \{m, \dots, \binom{m}{2}\}$ entails thus that $f^*(R,a)\geq f^*(\bar{R},x)$. We use this fact to construct a new profile $R^*$ where $f^*(R^*, x) \leq f^*(R,a)\le \beta$ for every $x\in A$. Let every voter $j\in N$ report the alternatives $x_j$, $y_j$ for which he is $0$-randomly dictatorial as his two best alternatives. Furthermore, distribute all other alternatives such that no alternative is ranked third by more than $m-2$ voters. This is possible as there are $m\geq 3$ alternatives and $\frac{m(m-1)}{2}$ voters. It follows from the construction that $r^*(x,R^*)_{k}\geq r^*(a,R)_k$ for every $k\in \{m, \dots, \binom{m}{2}\}$ and every $x\in A$. Hence, we derive that $f^*(R^*,x)\leq f^*(R,a)\leq \beta$ for every $x\in A$. If $\beta<\frac{1}{m}$, this entails that $\sum_{x\in A} f^*(R^*,x) <1$, a contradiction. Thus, $f^*$ cannot satisfy $\beta$-\emph{ex post} efficiency for $\beta<\frac{1}{m}$, and thus, $f$ violates this axiom, too. This show that there exists no $0$-randomly dictatorial SDS that can be represented as a mixture of unilaterals and that satisfies $\beta$-\emph{ex post} efficiency for $\beta < \frac{1}{m}$ when $m \ge 3$.
\end{proof}

Finally, we use \Cref{lem:dupleExPost} and \Cref{lem:expostUnilateral} to prove that there are no $0$-randomly dictatorial SDSs that satisfy $\beta$-\emph{ex post} efficiency for $\beta<\frac{1}{m}$. 

\expost*
\begin{proof}
	Let $f$ denote a strategyproof SDS for $n$ voters and $m\geq 3$ alternatives that is $0$-randomly dictatorial. Our argument focuses mainly on the profiles $R^{x,y}$, in which all voters report $x$ as their best choice and $y$ as their second best choice. The reason for this is that if $f(R,y)> \beta$ for some profile $R$ in which $y$ is Pareto-dominated by $x$, then $f(R^{x,y},y)>\beta$. This is a direct consequence of strategyproofness as we can transform $R$ into $R^{x,y}$ by reinforcing $x$ and $y$. Hence, non-perversity implies that $f(R^{x,y},y)\geq f(R,y)>\beta$. Moreover, localizedness entails that the order of the alternatives $z\in A\setminus \{x,y\}$ in $R^{x,y}$ is not important as it does not affect the probabilities of $x$ and $y$.
	
	Next, we use \Cref{thm:Gibbard1} to represent $f$ as mixture of duples and unilaterals, i.e, $f=\lambda f_\mathit{uni} + (1-\lambda) f_\mathit{duple}$, where $\lambda\in [0,1]$, $f_\mathit{uni}$ is a mixture of unilaterals, and $f_\mathit{duple}$ is a mixture of duples. While \Cref{lem:dupleExPost} and \Cref{lem:expostUnilateral} imply that $f_\mathit{uni}$ and $f_\mathit{duple}$ are not $\beta$-\emph{ex post} efficient for $\beta<\frac{1}{m}$, this does not imply that $f$ violates $\beta$-efficiency for $\beta<\frac{1}{m}$, too. The reason for this is that $f_\mathit{uni}$ and $f_\mathit{duple}$ may violate $\beta$-\emph{ex post} efficiency for different profiles or alternatives. We solve this problem by constructing a strategyproof SDS $f^*=\lambda f_\mathit{uni}^*+(1-\lambda) f_\mathit{duple}^*$ that is $0$-randomly dictatorial and $\beta$-\emph{ex post} efficient for the same $\beta$ as $f$, and for which $f_\mathit{uni}^*$ and $f_\mathit{duple}^*$ denote mixtures of unilaterals and duples such that $f^{*}_\mathit{uni}(R^{x,y},y)=f^*_\mathit{uni}(R^{\tau(x), \tau(y)}, \tau(y))$ and $f^{*}_\mathit{duple}(R^{x,y},y)=f^*_\mathit{duple}(R^{\tau(x), \tau(y)}, \tau(y))$ for all permutations $\tau:A\rightarrow A$. 
	
	For this construction, we define $f^{\tau}$ as $f^\tau(R,x)=f(\tau(R),\tau(x))$ for every permutation $\tau:A\rightarrow A$. We construct the SDS $f^*$ for $m!n$ voters as follows: we partition the electorate in $m!$ sets $N_k$ with $|N_k|=n$ and associate each of these sets with a different permutation $\tau_k:A\rightarrow A$. Then, we choose one of these sets $N_k$ uniformly at random and consider from now on only the preference profile $R_{N_k}$ defined by the voters in $N_k$. Finally, return $f^{\tau_k}(R_{N_k})$, where $\tau_k$ denotes the permutation associated with $N_k$. More formally, $f^*(R)=\frac{1}{m!}\sum_{k=1}^{m!}f^{\tau_k}(R_{N_k})$.
	
	First, note that $f^*$ is $0$-randomly dictatorial because of \Cref{lem:zerodict}. Since $f$ is a $0$-randomly dictatorial, there is for every voter $i$ a profile $R$ and alternatives $x,y$ such that voter $i$ prefers $x$ the most in $R$ and $y$ the second most, and $f(R,y)=f(R^{i:yx},y)$. Consequently, there are such profiles and alternatives for every voter in each SDS $f^\tau$. Finally, we derive that such profiles and alternatives exist also for $f^*$. For a voter $i\in N_k$, the corresponding alternatives $x,y$ and the preferences of the voters in $N_k$ are the same as for $f^{\tau_k}$. The preferences of the remaining voters do not matter. If $f^*$ does not choose $N_k$ in the first step, the preferences of voter $i$ do not matter, and if $f^*$ chooses $N_k$, it only computes $f^{\tau_k}(R_{N_k})$. Hence, if voter $i$ now swaps $x$ and $y$, the outcome of $f^*$ does not change as the outcome of $f^{\tau_k}$ does not change. Consequently, \Cref{lem:zerodict} implies that $f^*$ is $0$-randomly dictatorial. 
	
	Next, observe that $f^*(R)=\frac{1}{m!}\sum_{k=1}^{m!}f^{\tau_k}(R_{N_k})$ is strategyproof as it is a mixture of strategyproof SDSs. In particular, we can interpret each term $f^{\tau_k}(R_{N_k})$ as SDS for $m!n$ voters that ignores the preferences of the voters in $N\setminus N_k$. It follows immediately from this interpretation that $f^*$ is strategyproof because all $f^{\tau_k}$ are strategyproof. Hence, we can use \Cref{thm:Gibbard1} to represent $f^*$ as $f^*=\lambda f^*_\mathit{uni} + (1-\lambda) f^*_\mathit{duple}$, where $f^*_\mathit{uni}$ is a mixture of unilaterals and $f^*_\mathit{duple}$ is a mixture of duples. In more detail, the following equation shows that $f^*_\mathit{uni}(R)=\frac{1}{m!}\sum_{k=1}^{m!} f^{\tau_k}_\mathit{uni}(R_{N_k})$ and $f^*_\mathit{duple}(R)=\frac{1}{m!}\sum_{k=1}^{m!} f^{\tau_k}_\mathit{duple}(R_{N_k})$, where $f^{\tau}_\mathit{uni}$ and $f^{\tau}_\mathit{duple}$ are defined analogously to $f^{\tau}$. 
	\begin{align*}
	f^*(R)&=\frac{1}{m!}\sum_{k=1}^{m!}f^{\tau_k}(R_{N_k})\\
	&= \frac{1}{m!}\sum_{k=1}^{m!}\lambda f^{\tau_k}_\mathit{uni}(R_{N_k}) + (1-\lambda )f^{\tau_k}_\mathit{duple}(R_{N_k})
	\\&=\lambda \frac{1}{m!} \sum_{k=1}^{m!} f^{\tau_k}_\mathit{uni}(R_{N_k}) + (1-\lambda)\frac{1}{m!} \sum_{k=1}^{m!} f^{\tau_k}_\mathit{duple}(R_{N_k})\\
	&=\lambda f^*_\mathit{uni}(R_{N_k}) + (1-\lambda ) f^*_\mathit{duple}(R_{N_k})
	\end{align*}
	
	Note that the definitions of $f^*_\mathit{uni}$ and $f^*_\mathit{duple}$ entail that $f^*_\mathit{uni}(R^{x,y},y)=f^*_\mathit{uni}(R^{\tau(x),\tau(y)},\tau(y))$ and $f^*_\mathit{duple}(R^{x,y},y)=f^*_\mathit{duple}(R^{\tau(x),\tau(y)},\tau(y))$ for every permutation $\tau:A\rightarrow A$. For $f^*_\mathit{uni}$, this follows from the following equations and a symmetric argument holds for $f^*_\mathit{duple}$. 
	\begin{align*}
	f^*_\mathit{uni}(R^{x,y},y)&=\frac{1}{m!}\sum_{k=1}^{m!} f^{\tau_k}_\mathit{uni}(R_{N_k}^{x,y},y) \\
	&=\frac{1}{m!}\sum_{k=1}^{m!}f_\mathit{uni}(\tau_k(R_{N_k}^{x,y}),\tau_k(y)) \\
	&=\frac{1}{m!}\sum_{k=1}^{m!} f_\mathit{uni}(\tau_k(\rho(R_{N_k}^{x,y})),\tau_k(\rho(y))) \\
	&=f^*_\mathit{uni}(R^{\rho(x), \rho(y)}, \rho(y)))
	\end{align*} 
	
	The first two equations rely only on our definitions. The third equation follows because $\{\tau \circ \rho\colon \tau\in\Tau\}=\Tau=\{\tau_k\colon k \in \{1, \dots, m!\}\}$ for every permutation $\rho:A\rightarrow A$, where $\Tau$ is the set of all permutations on $A$. 
	
	Finally, we show that $f^*$ violates $\beta$-\emph{ex post} efficiency for every $\beta<\frac{1}{m}$, which entails that $f$ also violates this axiom. We use \Cref{lem:dupleExPost} and \Cref{lem:expostUnilateral} for this as these lemmas imply that $f^*_\mathit{duple}$ and $f^*_\mathit{uni}$ violate $\beta$-\emph{ex post} efficiency. Note for this that $f^*_\mathit{uni}$ is $0$-randomly dictatorial as otherwise, $f^*$ cannot be $0$-randomly dictatorial. Hence, there are profiles $R^{1}$ and $R^{2}$, and alternatives $x_1$, $y_1$, $x_2$, and $y_2$ such that $x_i$ Pareto-dominates $y_i$ in $R^i$ for $i\in \{1,2\}$, $f^*_\mathit{uni}(R^1,y_1)\geq\frac{1}{m}$, and $f^*_\mathit{duple}(R^2,y_2)\geq\frac{1}{m}$. Hence, we derive from strategyproofness that $f^*_\mathit{uni}(R^{x_1,y_1},y_1)\geq\frac{1}{m}$ and $f^*_\mathit{duple}(R^{x_2,y_2},y_2)\geq\frac{1}{m}$. Finally, it follows from the symmetry of $f^{*}_\mathit{uni}$ and $f^*_\mathit{duple}$ with respect to the profiles $R^{x,y}$ that $f^*_\mathit{uni}(R^{x,y},y)\geq\frac{1}{m}$ and $f^*_\mathit{duple}(R^{x,y},y)\geq\frac{1}{m}$ for all alternatives $x,y\in A$. Consequently, we conclude that $f^*(R^{x,y},y)=\lambda f^*_\mathit{uni}(R^{x,y},y)+(1-\lambda)f^*_\mathit{duple}(R^{x,y},y)\geq\frac{1}{m}$ for all $x,y\in A$. This means that $f^*$ and therefore also $f$ violate $\beta$-\emph{ex post} efficiency for every $\beta<\frac{1}{m}$.
\end{proof}

As last result, we discuss the proof of \Cref{thm:contrd}. 

\contrd*
\begin{proof}
	
	Just as for \Cref{thm:Cond}, we need to show two claims: on the one hand, there is for every $\epsilon\in [0,1]$ no strategyproof and $\frac{1 - \epsilon}{m}$-\emph{ex post} efficient SDS that is $\gamma$-randomly dictatorial for $\gamma<\epsilon$. On the other hand, we need to prove that every strategyproof and $\epsilon$-randomly dictatorial SDS that satisfies anonymity, neutrality, and $\frac{1 - \epsilon}{m}$-\emph{ex post} efficiency is a mixture of the uniform random dictatorship and the uniform lottery.\bigskip
	
	\textbf{Claim 1: For all $\epsilon\in[0,1]$, every strategyproof and $\frac{1 - \epsilon}{m}$-\emph{ex post} efficient SDS is $\gamma$-randomly dictatorial for $\gamma \ge \epsilon$.}
	
	 Consider an arbitrary SDS $f$ that is strategyproof and $\frac{1 - \epsilon}{m}$-\emph{ex post} efficient for some $\epsilon\in [0,1]$. By the defintion of $\gamma$-randomly dictatorial SDSs, there is a maximal $\gamma\in[0,1]$ such that $f$ can be represented as $f = \gamma d + (1 - \gamma) g$, where $d$ is a random dictatorship and $g$ is another strategyproof SDS. We need to show that $\gamma\geq \epsilon$. First, note that if $\gamma=1$, this is trivially the case since $\epsilon\in[0,1]$. On the other hand, if $\gamma<1$, the maximality of $\gamma$ entails that the SDS $g$ is $0$-randomly dictatorial. Hence, \Cref{lem:expost} shows that $g$ is at best $\frac{1}{m}$-\emph{ex post} efficient, i.e, there is a profile $R$ with a Pareto-dominated alternative $x$ such that $g(R,x)\geq \frac{1}{m}$. Since $f$ is $\frac{1-\epsilon}{m}$-\emph{ex post} efficient, we derive therefore the following inequality.
	\begin{equation*}
	\frac{1-\epsilon}{m}\geq f(R, x) = \gamma d(R, x) + (1 - \gamma) g(R, x) \ge \frac{1 - \gamma}{m}
	\end{equation*}
	
	This inequality is equivalent to $\epsilon\leq\gamma$ and therefore proves the claim.\bigskip
	
	\textbf{Claim 2: For all $\epsilon\in[0,1]$, every strategyproof and $\epsilon$-randomly dictatorial SDS that satisfies anonymity, neutrality, and $\frac{1-\epsilon}{m}$-\emph{ex post} efficiency is a mixture of the uniform random dictatorship and the uniform lottery.}
	
	Consider an arbitrary $\epsilon\in[0,1]$ and let $f$ denote an SDS for $m\geq 4$ alternatives that satisfies all axioms listed above. In particular, $f$ is $\epsilon$-randomly dictatorial and therefore, it can be represented a $f=\epsilon d + (1-\epsilon) g$, where $d$ is a random dictatorship and $g$ another strategyproof SDS. As first step, we show that $d$ needs to be the uniform random dictatorship. Note for this that anonymity implies that the values $\gamma_1, \dots, \gamma_n$ introduced in \Cref{lem:zerodict} are for all voters equal, i.e., $\gamma_i=\gamma_j$ for all $i,j\in N$. A close inspection of the proof of \Cref{lem:zerodict} reveals therefore that $d$ needs to be the uniform random dictatorship because we show for this lemma that, given the values $\gamma_i$, $f$ can be represented as $f=\sum_{i\in N} \gamma_i + (1-\sum_{i\in N}\gamma_i)g$. Here, $d_i$ denotes the dictatorial SDS of voter $i$. In particular, this means that $f$ is the uniform random dictatorship if $\epsilon=1$, which shows that our claim holds in this case.
	 
	Next, assume that $\epsilon<1$. In this case, the maximality of $\epsilon$ implies that the SDS $g$ is $0$-randomly dictatorial. Furthermore, $g$ needs to satisfy $\frac{1}{m}$-\emph{ex post} efficiency as otherwise, there is a profile $R$ and an alternative $x$ such that $f(R,x)=\epsilon d(R,x)+(1-\epsilon)g(R,x)>\frac{1-\epsilon}{m}$. This contradicts, however, the assumption that $f$ is $\frac{1-\epsilon}{m}$-\emph{ex post} efficient. As last point on $g$, observe that it is also anonymous and neutral as both $d$ and $f$ satisfy these axioms. 
	 
	Since $g$ is an anonymous, neutral, and strategyproof SDS, we can use \Cref{thm:Barbera} to represent $g$ as mixture of a point voting SDS $g_\mathit{point}$ and a supporting size SDS $g_\mathit{sup}$. These two SDSs are $0$-randomly dictatorial because $g$ satisfies this axiom. Furthermore, neither $g_\mathit{point}$ nor $g_\mathit{sup}$ can satisfy $\beta$-\emph{ex post} efficiency for $\beta<\frac{1}{m}$ because of \Cref{lem:expost}. Since $g$ is $\frac{1}{m}$-\emph{ex post} efficient, it follows from this observation that both $g_\mathit{point}$ and $g_\mathit{sup}$ need to satisfy this axiom, too. We show next that this implies that both $g_\mathit{point}$ and $g_\mathit{sup}$ always return the uniform lottery. 
	 
	First, consider $g_\mathit{point}$ and let $(a_1, \dots, a_m)$ denote its scoring vector. Our goal is to show that $a_1\leq \frac{1}{mn}$ because this implies that $g_\mathit{point}$ is the uniform lottery. This claim follows from the definition of scoring rules which requires that $a_1\geq a_2\geq\dots\geq a_m\geq 0$ and $\sum_{i=1}^m a_i=\frac{1}{n}$. If $a_1\leq \frac{1}{mn}$, this is only possible when $a_i=\frac{1}{mn}$ for all $i\in \{1,\dots, m\}$, which is the scoring vector of the uniform lottery. For showing that $a_1\leq \frac{1}{mn}$, note first that $a_1=a_2$ follows from \Cref{lem:zerodict} since $g_\mathit{point}$ is $0$-randomly dictatorial. Next, consider the profile $R$ in which all voters rank $a$ first and $b$ second. It follows from the definition of point voting SDSs that $g_\mathit{point}(R,b)=na_2$ and from $\frac{1}{m}$-\emph{ex post} efficiency that $g_\mathit{point}(R,b)\leq \frac{1}{m}$. Hence, we conclude that $a_1=a_2\leq \frac{1}{mn}$, which proves that $g_\mathit{point}$ is the uniform lottery. 
	 
	As last step, consider the supporting size SDS $g_\mathit{sup}$ and let $(b_n, \dots, b_0)$ denote its defining vector. Our goal is to show that $b_n\leq \frac{1}{m(m-1)}$ because the definition of supporting size SDSs entails then that $g_\mathit{sup}$ is the uniform lottery. In more detail, it must hold that $b_n\geq b_{n-1}\geq \dots\geq b_0$ and $b_{n-i}+b_{i}=\frac{2}{m(m-1)}$ for all $i\in \{0,\dots, n\}$. If $b_n\leq \frac{1}{m(m-1)}$, the conditions can only be satisfied when $b_i=\frac{1}{m(m-1)}$ for all $i\in \{0, \dots, n\}$, which is the scoring vector of the uniform lottery. It remains to show that $b_n\leq \frac{1}{m(m-1)}$. Consider for this the profile $R$ in which all voters agree that $a$ is the best and $b$ the second best alternative. It follows from the definition of supporting size SDSs that $g_\mathit{sup}(R,b)=(m-2)b_n + b_0=(m-3)b_n + \frac{2}{m(m-1)}$ and $\frac{1}{m}$-\emph{ex post} efficiency requires that $g_\mathit{sup}(R,b)\leq \frac{1}{m}$. Combining these two observations and solving for $b_n$ shows that $b_n\leq \frac{1}{m(m-1)}$ if $m \ge 4$, which proves that $g_{\mathit{sup}}$ is also the uniform lottery.
	 
	Since both $g_\mathit{point}$ and $g_\mathit{sup}$ need to be the uniform lottery, it follows also that $g$ itself is the uniform lottery. Thus, $f$ is indeed a mixture of the uniform random dictatorship and the uniform lottery.
\end{proof}

\end{document}